\documentclass[10pt,letterpaper,runningheads]{myllncs}

\pdfoutput=1 

\usepackage[english]{babel}

\usepackage{amssymb,amsmath} 
\usepackage{amsfonts,amssymb,mathtools}
\usepackage[boxruled,linesnumbered]{algorithm2e}
\usepackage{txfonts}
\usepackage{comment}
\usepackage{latexsym}
\usepackage{url}
\usepackage{fullpage}
\usepackage{multirow}

\usepackage[numbers]{natbib}

\def\K{{\mathbb K}} 
\newcommand{\softO}[1]{\mathchoice{\tilde{O}\left(#1\right)}{\tilde{O}(#1)}{\tilde{O}(#1)}{\tilde{O}(#1)}}

\newcommand{\res}{\text{\rm Res}}
\newcommand{\rev}{\text{\rm rev}}
\newcommand{\trsp}[1]{{#1}^\mathsf{T}}

\newcommand{\Sx}{{\ensuremath{S_{\!x}}}} 
\newcommand{\Sxk}[1]{{\ensuremath{S_{\!x}^{\!(#1)}}}} 
 
\newcommand{\Sy}{{\ensuremath{S_{\!y}}}} 
\newcommand{\Syk}[1]{{\ensuremath{S_{\!y}^{\!(#1)}}}} 
\newcommand{\Syp}{{\ensuremath{S_{\!y}\!\!'}}} 
\newcommand{\Sxp}{{\ensuremath{S_{\!x}\!\!'}}}

\newcommand{\Tx}{{\ensuremath{T_{\!x}}}}

\newcommand{\SA}{{\ensuremath{S_{\!A}}}} 
\newcommand{\SAe}{{\ensuremath{S^*_{\!A}}}} 
 
\newcommand{\SRe}{{\ensuremath{S^*_{\!R}}}}

\newcommand{\Al}{{\ensuremath{\mathbb A}}}

\makeatletter
\newcommand*{\rem}{%
  \nonscript\mskip-\medmuskip\mkern5mu%
  \mathbin{\operator@font rem}\penalty900\mkern5mu%
  \nonscript\mskip-\medmuskip
}
\makeatother

\usepackage{color,svgcolor}
\usepackage[plainpages=true,bookmarks=false]{hyperref}
\makeatletter
\hypersetup{
pdftitle={\@title},
breaklinks=true,
colorlinks=true,
linkcolor=darkred,
citecolor=green,
urlcolor=green,
}
\makeatother

\title{Elimination ideal and bivariate resultant over finite fields
}

\author{
{Gilles Villard} 
}

\institute{  
{CNRS, U. Lyon, Inria, ENS de Lyon, UCBL, Laboratoire LIP UMR5668, France}	
}

\usepackage[capitalize]{cleveref}

\crefname{Item}{Step}{Steps}

\numberwithin{theorem}{section}
\numberwithin{proposition}{section}
\numberwithin{lemma}{section}
\numberwithin{remark}{section}
\numberwithin{corollary}{section}
\numberwithin{example}{section}

\begin{document}

\maketitle 

\thispagestyle{empty}

\begin{abstract}
A new algorithm is presented for computing the largest degree invariant factor of the Sylvester 
matrix~(with respect either to $x$ or $y$)  associated to two polynomials $a$ and $b$ in $\mathbb F_q[x,y]$ which have no non-trivial common divisors. The algorithm is randomized 
of the Monte Carlo type and requires $O((de)^{1+\epsilon}\log(q) ^{1+o(1)})$ bit operations, 
where $d$ an $e$ respectively bound the input degrees in $x$ and in $y$. It follows that 
the same complexity estimate is valid for computing: a generator of the elimination 
ideal $\langle a,b \rangle \cap \mathbb F_q[x]$ (or $\mathbb F_q[y]$), as soon 
as the polynomial system $a=b=0$ has not roots at infinity; the resultant of $a$ and $b$ 
when they are sufficiently generic, especially so that the Sylvester matrix has a unique non-trivial invariant factor.  
Our approach is to use the reduction of the problem to a problem of minimal polynomial in the 
 quotient algebra $\mathbb F_q[x,y]/\langle a,b \rangle$. 
 By proposing a new method based on structured polynomial matrix division for computing with the elements 
 in the quotient, we manage to improve the best known complexity bounds. 
\end{abstract}

%
%

\section{Introduction} \label{sec:intro}

Given two polynomials $a,b \in \K[x,y]$, where $\K$ is a commutative field, their resultant $\res _y(a,b)$ with respect to $y$ is the determinant of the associated Sylvester matrix $\Sy$ 
over $\K[x]$~\cite[Ch.\,6]{GaGe99}. 
Computing this determinant in 
quasi-linear time with respect to the input/output size 
is still beyond our reach in the general case. 

\smallskip

In this paper we consider the relaxed problem which is to compute the last (of largest degree) invariant factor   of $\Sy$, in the case of a finite field~$\K=\mathbb F_q$ with $q$ elements. We consider $a$ and $b$  of 
 $x$-degree at most~$d$ and $y$-degree at most $e$ in $\mathbb F_q[x,y]$, having no non-trivial common divisors. For any $\epsilon >0$, 
there exist a randomized Monte Carlo algorithm which solves the problem using a quasi-linear number of 
$O((de)^{1+\epsilon}\log(q) ^{1+o(1)})$ bit operations. 

\smallskip

The last invariant factor of~$\Sy$ is a specific divisor of the resultant. 
If the polynomial system $a=b=0$ has no roots at infinity with respect to $y$ (the $y$-leading coefficients of $a$ and $b$ are coprime), then 
it  gives central informations on the affine solutions. It is indeed 
a generator of the elimination ideal $\langle a,b\rangle \cap \K[x]$~\cite[Ch.\,2]{CoLiOSh05}. We also have, in particular, the fact that this invariant factor gives
the resultant when $a$ and $b$ are sufficiently generic~(\cref{sec:resultant}). 
(Genericity is considered in the Zariski sense: a property is generic if it holds except on a hypersurface of the parameter space.)

\smallskip 

Our approach over finite fields is inspired by and goes further than the major steps taken with: the change of order algorithm of Poteaux and Schost for triangular sets and radical ideals~\cite{PS13}; the algorithm of van der Hoeven and Lecerf, 
which computes the resultant of generic polynomials with respect to the total degree~\cite{HoeLec21a}. In the bivariate case, both these works provides solutions in quasi-linear expected time in the input/output size for the first time (\cite{PS13} treats general multivariate cases). They are part of the 
 same long line of research which reduces elimination problems to linear algebra [\citealp{Laz81}; \citealp[Sec.\,2.4 \& 3.6]{CoLiOSh05}], and especially to the computation of minimal 
polynomials in quotient algebras~\cite{lazard1992,Shoup99}. It is this path that we are pursuing.

\smallskip 
\noindent 
{\em The role of minimal polynomials.}
 Let $I=\langle a,b\rangle$ be the (zero-dimensional) ideal generated by $a$ 
and $b$ in~$\K[x,y]$, and $\Al=\K[x,y]/I$ be the associated quotient algebra. 
We remind in  \cref{sec:ideals} that the last invariant invariant factor of the Sylvester matrix $\Sy$ can be computed 
as the minimal polynomial $\mu$ of the multiplication by $x$ in $\Al$, under 
the condition of absence of roots at infinity~\cite{Lazard85}. 
This is how we proceed. The condition on the behaviour at infinity is 
met for slightly modified polynomials not preventing us from computing the target invariant 
factor (\cref{sec:shifts}).

For efficiency, the minimal polynomial problem is itself reduced to a power projection problem~\cite[Sec.\,6]{Kal00-2}
(a more complete list of references is given later in this introduction).
Given a  linear form $\ell$ in the dual of~$\Al$ over~$\K$, the minimal polynomial in $\Al$ is computed as the one 
of the linearly generated sequence~$\{\ell(x^i \bmod I)\}_{i\geq 0}$ over~$\K$.
The application of a random linear form preserves the recursion which is sought in $\Al$~\cite{Wie86} (\cref{sec:invfact}). 
As observed by Shoup~\cite{Shoup94}, the power projection problem is dual to the modular 
composition problem~\cite{BK78}.  We finaly rely on Kedlaya and Umans' approach to address those two latter issues~\cite{KU11} in quasi-linear time over finite fields. As we will now see, this is made possible by a new algorithm we propose for arithmetic operations modulo the ideal.   

\smallskip 

\noindent 
{\em First result.}
One of the bottlenecks in above strategy is to perform arithmetic operations in $\Al$~\cite{Hoeven17}, even if only to compute 
the multiplication of two polynomials or the powers of $x$ that need to be projected modulo the ideal $I$.  This is where a  main aspect of our contribution lies. In~\cite{PS13}, the special case of triangular sets is considered. That is, in our context, when either $a$ or $b$ is univariate. On the other hand, the generic resultant algorithm of \cite{PS13} relies on Gr\" obner bases techniques, and the normal form algorithm modulo~$I$ of \cite{HoevenLarrieu2019}. 

We instead use polynomial matrix division~\cite[Sec. 6.3]{Kailath80}. 
Viewing a polynomial~$f$ in $\K[x,y]$ as a vector with entries in $\K[x]$, we reduce its $x$-degree using division by the polynomial Sylvester matrix~$\Sy$; 
let us also specify that we may need to construct a Sylvester matrix from multiples of $a$ and~$b$ if the dimensions do not 
match (\cref{sec:division}). By definition of the Sylvester matrix, the remainder of this division gives a new polynomial 
in the coset $f+I$. By means of a similar division after the swich of the roles of~$x$ and $y$, this leads 
to a normal form algorithm modulo $I$, up to a regularity assumption related to roots at infinity~(\cref{lem:rootsinfty} and \cref{prop:division}). 
This algorithm is algebraic  and deterministic for arbitrary fields. If $f$ has $x$-degree at most $\delta$ 
and   $y$-degree at most $\eta$, 
 then it uses 
$\softO{(d+\delta)(e+ \eta)}$ arithmetic operations (\cref{prop:division}). 
A Sylvester matrix is a Toeplitz-like matrix~\cite{BiPa94}. Our cost bound is based on fast structured matrix arithmetic which is discussed in \cref{sec:structmat}. In particular, 
the normal form algorithm allows  multiplication in $\K[x,y]/\langle a,b\rangle$ using 
$\softO{de}$ operations when the leading coefficients of $a$ and $b$ are sufficiently 
generic (\cref{lem:rootsinfty}). 
In the case of the total degree, for generic polynomials~$a,b$ with $\deg a \geq \deg b$, the algorithm of \cite{HoevenLarrieu2019}
costs $\softO{(\deg a)(\deg b)}$, after the precomputation of a concise Gr\"obner basis representation of the ideal using 
$\softO{(\deg  a)^2}$ operations. So in terms of their assumptions the two algorithms are complementary~(\cref{sec:polydiv}).

\smallskip 
\noindent
{\em Extension of Kedlaya and Umans' techniques for the power projections.}
As soon as the normal form algorithm is available, hence the arithmetic operations in $\Al$, it is possible to develop the general strategy
of Shoup~\cite{Shoup94,Shoup99} for the computation of modular power projections, coupled by duality with the algorithm 
of Kedlaya and Umans for modular composition~\cite{KU11} (in this latter reference, the case of a univariate ideal~$I$ in $\mathbb F_q[x]$ is treated). 
This is what has been generalized in both~\cite{PS13} 
and \cite{HoeLec21a}, with respective shapes of the ideal $I$ that we have seen above. 
We proceed in the same way, and integrate the new division algorithm into this overall process: (i)
reduction of  $f \in \mathbb F_q[x]$ modulo~$I$, which is considered as modular composition
according to $f(g(x)) \bmod I$ with $g=x$; modular composition relies on multivariate multipoint evaluation 
following \cite[Thm.\,3.1]{KU11}; (ii) using the transposition principle~\cite[Thm.\,13.20]{BCS97}, the power projections 
are obtained~(\cref{sec:Tdivision}). Since the degree of the resultant of $a$ and $b$ with respect to $y$ is at most~$2de$, 
it is sufficient to be able to compute (i) $f$ modulo $I$ for $\deg f < 4de$ and (ii) $\{\ell(x^i \bmod I)\}_{0\leq i< 4de}$, in order to deduce the minimal polynomial of the sequence (which is a divisor of the resultant). 
We establish in \cref{sec:composition} that (i) and (ii) can be perfomed within our target cost bound over a finite field using bit operations.

\smallskip 
\noindent
{\em Last invariant factor and elimination ideal $\langle a,b\rangle \cap \mathbb F_q[x]$.}
In the presence of roots at infinity, we use a random transformation of $a$ and $b$ into two other polynomials which 
meet the condition of regularity for computing normal forms, and still make it possible to obtain the initial last invariant factor. This is presented in \cref{sec:shifts}.  The general complexity bound for the computation of the last invariant factor is given in \cref{sec:invfact} from that of modular power projection.  As a consequence of Lazard's structure theorems for bivariate ideals~\cite{Lazard85}, the latter polynomial is 
a multiple of the minimal polynomial~$\mu$ of the multiplication by $x$ in $\Al$, and  
both polynomials coincide if the system $a=b=0$ has no roots at infinity (\cref{lem:laz}). 
Under this condition, what we have done so far allows in 
\cref{sec:resultant} to compute $\mu$, that is a generator of the elimination ideal~$I \cap \mathbb F_q[x]$.

\smallskip 
\noindent
{\em Comparison to previous work.} Given an arbitrary field, the bivariate resultant can be computed using 
$O(de^2)$ arithmetic operations~\cite[Chap.\,11]{GaGe99}. 

Over a finite field, 
the approach of \cite{HoeLec21a} allows quasi-linear bit cost; for generic polynomials with respect to the total degree, and any $\epsilon >0$, this leads to the complexity bound $O(((\deg a)(\deg b)\log q)^{1+\epsilon})+\softO{(\deg a)^2\log q}$ when $\deg a \geq \deg b$.
(The soft-$O$ notation $\softO{c}$ captures an additional logarithmic factor $O(c\log ^kc)$ for a positive $k$.)  
Our algorithm covers this case, in particular. Genericity ensures that there are no roots at infinity and a 
unique invariant factor (see \cref{sec:resultant}), and we obtain a comparable asymptotic bound.  
Considering degree conditions on the variables individually we treat a larger class of problems
and with weaker assumptions.
For polynomials of $x$-degree $d$ and $y$-degree~$e$, 
  we compute the resultant in quasi-linear  time when the Sylvester matrix $\Sy$ has a unique non-trivial invariant factor.

Now, in a general way, in all cases as soon as there are no roots at infinity, our approach allows to compute 
a generator of the elimination ideal $I\cap \mathbb F_q[x]$. This is treated in \cref{sec:resultant}. 
We are not aware of any previous method whose cost would be quasi-linear over finite fields under the same 
assumptions. The complexity of this problem is indeed related to that of the resultant
and bivariate lexicographic Gr\"obner bases~\cite{Dah22}. 
In particular, for $a$ and $b$ of total degree at most $d$, we arrive at the bound   
$O(d^{2+\epsilon}\log(q) ^{1+o(1)})$, while previous estimates are $\softO{d^3\log q}$~\cite{LMS13}.

 We use bit complexity. The bivariate resultant problem using an algebraic model of computation 
 is a harder problem. To our knowledge, a quasi-linear complexity bound is not achievable at this time. 
We may refer to [\citealp{Lec17}; \citealp{Vil18,PSV22}] and to the pointers found there. 
It it also important to note that
quasi-linear time algorithms are given for bivariate polynomial systems with integer coefficients 
in \cite{MS16b}.

\smallskip 
\noindent
{\em Minimal polynomials and power projections.} 
 We give here some additional references from which the results we use largely inherit. 
The adaptation of numerical matrix methods to the finite field setting has started with the solution of 
sparse linear systems in mind~\cite{COS86,Wie86}. These methods result in  projections of the powers of the involved matrix for computing its minimal polynomial, as evidenced by Wiedemann's approach~\cite{Wie86}.
The link is to be made with the use of power projections for computing minimal polynomials 
in quotient algebras, using the trace map in \cite{Thi89,RiBo91} and general projections 
in \cite{Shoup94,KaSh98,Shoup99}. 
(We have indeed a multiplication endomorphism in the quotient.)

The duality between the power projection problem and the modular composition one is observed 
in \cite{Shoup94}.  

In the context of polynomial system solving, for which the literature is vast,
 we may refer to the use of the trace map in \cite{Alonso96,Rou99,Gonzalez-Vega1999}, or of arbitrary linear forms  in \cite{BSS03}.
 Structured matrices and duality are applied to multivariate polynomial problems in \cite{MoPa2000}. 
Multivariate powers projections are considered in \cite{Shoup99,Kal00-2}, especially for minimal polynomials, and are exploited 
for the computation of special resultants in \cite{BFSS06}, and to the 
change of order of variables for triangular sets in~\cite{PaSc06}. 
The link between the change of ordering and linear algebra is also beneficial using power projections 
of a multiplication matrix in \cite{FaMo11,faugere2014sub}, and particularly in order to take advantage of 
sparsity~\cite{FaMo17}, which brings us back to Wiedemann's algorithm.

Following \cite{Lazard85}, the Sylvester matrix $\Sy$ (or $\Sx$) is a polynomial matrix that we manipulate as such. This may be seen as  
 working in a $\K[y]$-module rather than in a $\K$-vector space in order 
 to represent the quotient algebra $\Al$~\cite[Sec.\,3.10]{Jac85}, and implement the operations on its elements. A similar direction has been taken in \cite{BNS22} for a change or ordering 
 of Gr\"obner bases algorithm. 

 In linear algebra with implicitly represented matrices, an open problem 
 is to compute the characteristic polynomial in essentially the same 
time as for the minimal polynomial~[Sec.\,3]\cite{Kal00-2}. This applies in particular to sparse or structured matrices. 
The question of computing the bivariate resultant in essentially the same 
time as for the last invariant factor of the Sylvester matrix appears to be similar to Kaltofen's open problem.  

\smallskip 

\noindent
{\em Model of computation.}
The normal form algorithm for polynomials in $\K[x,y]$ modulo $\langle a,b\rangle$
and its transpose are  presented using an algebraic model (\cref{sec:division}), and work e.g. with 
computation trees~\cite[Sec.\,4.4]{BCS97}. Complexity bounds correspond to 
numbers of arithmetic operations performed in $\K$.

The application of Kedlaya and Uman's techniques in \cref{sec:composition} and threfore 
the last invariant factor computation in \cref{sec:invfact} rely on a RAM bit complexity model.  
We consider that arithmetic operations in $\mathbb F_q$ can be done in time $\softO{\log q}$, 
and that the RAM can produce a random element uniformly distributed in $\mathbb F_q$ with the same cost. 

\smallskip 

\noindent
{\em Notations.} Throughout the paper we consider two polynomials $a,b \in \K[x,y]$, of degrees $d_a$ and $d_b$ in $x$, and $e_a$ and~$e_b$ in $y$, respectively. We will use the notations 
$d=\max\{d_a,d_b\}$ and $e=\max\{e_a,e_b\}$. 
The associated Sylvester matrices with respect to~$x$ and $y$ are 
$\Sx \in \K[y]^{n_x \times n_x}$ and $\Sy \in \K[x]^{n_y \times n_y}$,  
with dimensions $n_x=d_a+d_b$ and $n_y=e_a+e_b$. 
The resultants $\res_x(a,b)\in \K[y]$ and $\res_y(a,b)\in \K[x]$, of $a$ and $b$ with respect to $x$ 
and $y$, are the respective determinants of $\Sx$ and $\Sy$~\cite[Chap.\,6]{GaGe99}. 
We assume that $a$ and $b$ have no non-trivial common divisors, hence both $\Sx$ and $\Sy$ are non-singular.
We focus on computations in relation to $\res_y(a,b)=\det \Sy$ (the conclusions would be unchanged in relation to $\res_x(a,b)$).

We use expressions such as ``$x$-degree'' or ``$y$-leading coefficient'' to indicate the variable 
which is concerned, and use $\deg_x$ and $\deg_y$ in formulas when bivariate polynomials are involved. Subscripts for example in $\K[x,y]_{<(d,n_y)}$ indicate degree bounds in $x$ and $y$, 
and $\K[x]_{d}^n$ is the set of polynomials of degree $d$.

We are often led to manipulate reversals of polynomials. For $k\geq 0$, we define 
the reversal of a polynomial $f\in \K[x]$ with respect to $k$ as $\rev _k (f)=x^kf(1/x)$; 
by default, if $k$ is not specified, the reversal is taken with respect to the degree of the polynomial.  This is generalized to polynomial matrices viewed as matrix polynomials, 
we mean with matrix coefficients.

The polynomials in $\K[x,y]$ are identified with the (column) vectors of their 
coefficients, using dimensions which will be clear from the context. 
For example, given $f=f_0(x)+f_1(x)y+\ldots f_d(x) y^d$ and $n\geq d+1$, $v_y(f)\in \K[x]^n$ denotes 
the vector 
 $\trsp{[0~\ldots ~0~ f_d~ \ldots ~ f_0 ]}$.

%
%

\section{Polynomial matrices, resultant and bivariate ideals} \label{sec:ideals}

We give the basic notions and results we need in the rest of the text concerning the relations between the resultant of two polynomials and the ideal they generate.   
 As univariate polynomial matrix, the Sylvester matrix $\Sy$ is unimodularly equivalent to a matrix 
${\text{\rm diag}}(s_1, \ldots, s_n)\in \K[x]^{n_y\times n_y}$ in Smith normal form, where~$s_n$ is the invariant factor of largest degree.
We are not able to always compute the resultant within the cost target. We are, however, able to compute the 
last invariant factor (\cref{cor:invfact}).

Using the structure theory of finitely generated modules, this last invariant factor can be seen as the minimal polynomial of a linear transformation in a finite dimensional $\K$-vector space~\cite[Sec.\,3.10]{Jac85}. 
Such a formalism has been exploited occasionally for the efficient computation of general matrix normal forms~\cite{villard1997fast,Sto00}. Concerning Sylvester matrices and in the broader context of polynomial system solution, this is related to the use 
of a multiplication map on a quotient algebra~\cite{Laz81}. 

Let $I=\langle a,b\rangle$ be the (zero-dimensional) ideal generated by $a$ and $b$ in $\K[x,y]$, and $\Al=\K[x,y]/I$ be the associated quotient algebra. We especially rely on the following results, which are immediate 
consequences of Lazard's theorem~\cite{Lazard85}.

\begin{lemma}[{\normalfont \cite[Thm.\,4]{Lazard85}}] \label{lem:laz}
The last invariant factor 
of $\Sy$ is a multiple of the minimal polynomial of the multiplication by $x$ in $\Al$, 
both polynomials coincide if the $y$-leading coefficients of $a$ and $b$ are coprime in $\K[x]$.  
\end{lemma}
\begin{proof} 
The last invariant factor is a multiple of the last diagonal entry $h\in \K[x]$ of the Hermite form, where the latter is lower triangular and obtained by unimodular column transformations. The polynomial $h$ 
is in $I \cap \K[x]$ (combinations 
of columns of $\Sy$ are seen as combinations of $a$ and $b$), which gives the first assertion.  When the leading coefficients are coprime, the divisibility property~(i) in \cite[Thm.\,4]{Lazard85} shows that the Hermite form 
of $\Sy$ can be brought to Smith form using unimodular (row) transformations, without modifying the diagonal. From (ii) in \cite[Thm.\,4]{Lazard85}, the last invariant factor is therefore an element of a reduced 
Gr\"obner 
basis of $I$, 
and as polynomial in $\K[x]$ it generates the elimination ideal $I \cap \K[x]$. 
\qed
\end{proof}

The condition on the leading coefficients of $a$ and $b$ in \cref{lem:laz} is the fact that the system $a=b=0$ has no roots at infinity with respect to $y$. In general, the resultant and the last invariant factor may have terms coming from both
the affine variety and the behaviour at infinity~\cite[Chap.\,3]{CoLiOSh05}.
To still be able to reduce the invariant factor computation to a minimal polynomial problem the assumption 
of \cref{lem:laz} will hold after a random modification of the input polynomials (see \cref{sec:shifts}).

The resultant can be deduced from \cref{lem:laz} in particular  when 
the Smith form of $\Sy$ has a unique non-trivial invariant factor and there are no roots at infinity. 
This corresponds to certain situations in which the ideal $I$ has a shape basis~\cite{GTZ88,BMMT94}.

\begin{lemma}[{\normalfont \cite[Thm.\,4]{Lazard85}}] \label{lem:shape}
The $y$-leading coefficients of $a$ and $b$ are coprime in $\K[x]$ and there exist two polynomials 
$\mu, \lambda \in \K[x]$ such that $I=\langle \mu(x), y-\lambda(x)\rangle$  if and only if, up to a non-zero element in $\K$, the resultant $res_y(a,b)$ 
is the minimal polynomial $\mu$ of the multiplication by $x$ in $\Al$. 
\end{lemma}
\begin{proof} 
From \cite[Thm.\,4]{Lazard85}, under the hypothesis $I=\langle \mu(x), y-\lambda(x)\rangle$ and using the coprimeness, we know that the Hermite  form of $\Sy$ has a unique non-trivial diagonal entry, which 
is $\mu$. Therefore, the latter is also the determinant of~$\Sy$, up to the normalization to a monic polynomial in the Hermite form. 

Conversely, the last element $h\in \K[x]$ of the diagonal of the Hermite form of $\Sy$ is in $I$ . Hence $h$ must be a multiple of $\mu$, and of 
the resultant by assumption. It follows that $h=\res_y(a,b)=\mu$, and all the other diagonal entries of the Hermite form are 
equal to $1$. This proves that the  $y$-leading coefficients of $a$ and $b$ are coprime since otherwise the 
first diagonal entry of the Hermite form would be a non-constant polynomial in $\K[x]$. Item 
(ii) \cite[Thm.\,4]{Lazard85} allows to conclude. 
\qed
\end{proof}

Concerning the links bewteen the resultant and the associated ideal, the reader may especially refer to \cite{CD21}, where a general multivariate version of \cref{lem:shape} is given.

\begin{example} \label{example1}
 The Sylvester matrix may have a unique non-trivial invariant factor (that our algorithm will compute) even though there are roots at infinity. With $\K=\mathbb F_2$, take $a=\left(x +1\right) y +x^{2}$ and $\left(x +1\right) y^{2}+y$. We have $I=\langle x^2,y\rangle$, and the Hermite normal form of $\Sy$ is 
 $$
\Sy U=\left[\begin{array}{ccc}
x +1 & 0 & 0 
\\
 1 & x +1 & 0 
\\
 0 & x^{2} & x^2\left(x +1\right) 
\end{array}\right],
 $$
 with $U$ unimodular. None of the arguments used for \cref{lem:laz,lem:shape} apply: the Hermite form cannot be brought to Smith form using unimodular row operations without modifying the diagonal (as used in the proof of \cref{lem:laz}), and the form is not either trivial (proof of \cref{lem:shape}). 
 The last invariant factor of $\Sy$ is $\res _y(a,b)=x^{2} \left(x +1\right)^{3}$. 
 \qed
\end{example}

We now characterize the existence of roots at infinity using column reducedness of polynomial matrices~\cite[Sec.\,6.3,\,p.384]{Kailath80}, which is used in next sections.  
Let $S$ be a matrix in $\K[x]^{n\times n}$  whose column $j$ has degree $d_j$. We call (column) leading (matrix) coefficient 
of $S$ the matrix in $\K^{n\times n}$ whose entry $(i,j)$ is the coefficient of degree $d_j$ of the entry $(i,j)$ of $S$.  We manipulate non-singular univariate polynomial matrices, and say that 
such a  matrix is column reduced if its leading coefficient is invertible. 

\begin{lemma}\label{lem:rootsinfty}
$\Sx$ is column reduced if and only if, the $y$-leading coefficients of $a$ and $b$ are relatively prime and at least one of latter polynomials in $\K[x]$ has maximal degree $d_a$ or $d_b$, respectively. 
\end{lemma}
\begin{proof}
Let $s,t \in \K[x]$ 
 be the $y$-leading coefficients of $a,b$, with respective degrees $d_s$ and $d_t$.
The columns of the leading coefficient of $\Sx$ are given by the vectors in $\K^{d_a+d_b}$ associated to 
$$x^{d_b-1}s,x^{d_b-2}s,\ldots,s,x^{d_a-1}t,x^{d_a-2}t,\ldots,t.$$
If $\Sx$ is column reduced then the first row of its leading matrix is non-zero and either $d_s = d_a$ or $d_t=d_b$. Let's say that~$d_s=d_a$ (up to a column permutation). The leading coefficient of $\Sx$ is therefore  
given by 
\begin{equation}\label{eq:lemsylv}
x^{d_b-1}s,x^{d_b-2}s,\ldots,x^{d_t}s,x^{d_t-1}s,x^{d_t-2}s,\ldots,s,x^{d_s-1}t,x^{d_s-2}t,\ldots,t,
\end{equation}
and we see that its rank is that of the Sylvester matrix associated to $s$ and $t$ since the latter is given by  
$$x^{d_t-1}s,x^{d_t-2}s,\ldots,s,x^{d_s-1}t,x^{d_s-2}t,\ldots,t.$$
Conversely, from the independence of the vectors in \cref{eq:lemsylv} we obtain the column reducedness of $\Sx$. 
\qed
\end{proof}

%
%

\section{Bivariate polynomial division} \label{sec:division}

In this section we propose a normal form algorithm for bivariate polynomials modulo the ideal $I=\langle a,b\rangle$. 
The algorithm relies on matrix polynomial division. Bivariate polynomials in $\K[x,y]$ are viewed as 
univariate polynomial vectors alternately over $\K[x]$ and $\K[y]$, dividing such a vector by $\Sy$ or $\Sx$, is indeed quivalent to 
reducing the associated polynomial modulo the ideal. Sylvester matrices are Toeplitz-like matrices, 
we first recall in \cref{sec:structmat} how operations on matrices in this class can be performed 
taking into account their structure~\cite{BiPa94,pan01}. We then study the division with remainder of a polynomial vector by $\Sy$ or $\Sx$ in 
\cref{subsec:vecdiv}. In order to be able to define a normal form and perform the division efficiently, we rely on  
a regularity assumption on leading coefficient matrices: we suppose that $\Sx$ and $\Sy$ are column reduced. This assumption is ultimately harmless for computing the last invariant factor~(\cref{sec:shifts}). 

In \cref{sec:polydiv} we present the normal form algorithm. We keep the same notations as before  
for the degrees of $a$ and $b$, and the dimensions of the matrices; especially, $d$ is the maximum 
degree in $x$ and $\Sy$ is $n_y\times n_y$. 
 Given a  polynomial~$f\in \K[x,y]$, we show how to compute a unique polynomial $\hat f \in \K[x,y]_{<(d,n_y)}$,  that we denote by $\hat f = f \rem I$, such 
that $f - \hat f \in I$ (\cref{prop:division}).
Uniqueness is ensured using a properness property provided by the polynomial matrix division. 
The construction is a $\K$-linear map that sends $f$ to $\hat f$ whose $y$-coefficients are given by the entries of a vector $v_y(\hat f) \in \K[x]_{< d}^{n_y}$ such that 
$\Sy^{\!\!-1}v_y(\hat f)$ is strictly proper (tends to zero when $x$ tends to infinity), 
see \cref{eq:divvec}.
This allows us to represent the elements in $\Al$ 
by normal forms. The transpose algorithm, which computes corresponding power projections, is derived in 
\cref{sec:Tdivision}

With  $\deg_x a = d_a$, $\deg_x b = d_b$, $\deg_y a = e_a$, and $\deg_y b = e_b$,
 the quotient algebra $\Al$ has dimension at most~$d_a e_b+ d_be_a$. In order to represent its elements, the quotient is embedded in the space $\K[x]_{< d}^{n_y}$ of 
dimension $$dn_y=\max\{d_a,d_b\}(e_a+e_b)$$ 
which can therefore be slightly larger (\cref{ex:notproper}).


\subsection{Structured matrix arithmetic} \label{sec:structmat}

The normal form algorithm exploits the fact that Sylvester matrices are structured. The class of structure that we are facing 
is the one of {Toeplitz-like polynomial matrices} which are commonly handled using the notion of {displacement rank}~\cite{KKM79}.  
The notion allows to have a concise matrix representation 
through which matrix arithmetic can be implemented efficiently~\cite{BiPa94,pan01}. 

Given by the polynomials $a$ and $b$, $\Sx$ and $\Sy$ are represented using $O(de)$ elements of $\K$.
The division algorithm  requires to solve associated linear systems and uses matrix inversion with truncated power series entries. 
We consider that polynomial Sylvester matrices and their inverses are represented using their concise Toeplitz-like representations~\cite{BiPa94}. This is obtained for example by extension of the 
$\Sigma LU$ form defined over fields~\cite{Kal94}, to  polynomials or truncated power 
series~\cite[Sec.\,3]{PSV22}.

Multiplying an $n\times n$ polynomial Sylvester matrix of degree $d$ by a polynomial vector of 
degree at most $l$ over~$\K[x]$, can be done using $\softO{n(d+l)}$ arithmetic operations in $\K$~\cite{BiPa94}. This cost bound is valid for the same type of multiplication using instead the inverse of the matrix modulo $x^l$ when it exists. 
If $T\in \K^{n\times n}$ is a non-singular Sylvester matrix and~$v\in \K^n$, then the 
linear system $T^{-1}v$ can be solved using~$\softO{n}$ arithmetic operations. This is obtained by combining an inversion formula for the Sylvester matrix \cite{La1992}, and matrix Pad\'e approximation \cite{BeLa94} 
 (see also 
\cite[Chap.\,2, Sec.\,9]{BiPa94} and \cite[Sec.\,5]{Vil18}). 
The declination of this is applied in \cref{subsec:vecdiv} over truncated 
power series modulo~$x^l$. Let~$S\in \K[x]_{d}^{n\times n}$ be a polynomial Sylvester matrix 
such that $\det S(0) \neq 0$, and consider a vector $v\in \K[x]^{n}$ of degree at most $l$. The system $S^{-1}v$ can be solved modulo~$x^l$
using~$\softO{n(d+l)}$ arithmetic operations. 
From~\cite[Prop.\,5.1]{Vil18}, the matrix inverse modulo $x^l$ can itself be computed (with concise representation) within the same cost bound.


\subsection{Matrix and bivariate polynomial division} \label{subsec:vecdiv}

Consider $S$ in $\K[x]^{n\times n}$, non-singular of degree $d$. For any vector $v\in \K[x]^n$, we know from 
\cite[Thm.\,6.3-15, p.\,389]{Kailath80} that there exist 
unique~$w, \hat v \in \K[x]^n$ such that 
\begin{equation}\label{eq:divvec}
v = Sw + \hat v,
\end{equation}
and $S^{-1}\hat v$ is strictly proper.  From \cite[Thm.\,6.3-10, p.\,383]{Kailath80} we further have that 
the polynomial remainder vector $\hat v$ has degree less $d$; note however that uniqueness is ensured by properness and not by the latter 
degree property~(\cref{ex:notproper}). 

The following will be applied to both $\Sx$ and $\Sy$, hence we take a general notation $S$ for the statement. 
We propose a structured matrix polynomial adaptation of the Cook-Sieveking-Kung algorithm for (scalar) polynomial division with remainder, about which the reader may refer to \cite[Sec.\,9.1]{GaGe99}.

\begin{lemma} \label{lem:matdiv}
Let $S \in \K[x]^{n\times n}$ be a  Sylvester matrix of degree $d$, and assume that $S$ is column reduced.   
Consider 
a vector $v\in \K[x]^n$ of degree at most $l$. The unique remainder $\hat v$ of the division of $v$ par $S$ as in \cref{eq:divvec} can be computed using~$\softO{n(d+l)}$ arithmetic operations in $\K$. 
\end{lemma}
\begin{proof}
Consider that  $S$ is associated to two polynomials $a,b\in K[x,y]$ as previously, such that $S=\Sy$ and we have $e_1$ columns of degree $d_1$ and $e_2$ columns of degree $d_2$. 
Up to row and column permutations we assume that $d=\max\{d_1,d_2\}=d_1$. 

We first treat the case $d=d_1=d_2$. All the columns of $S$ have the same degree, hence since $S$ is column reduced it is also 
row reduced (use the definition given before \cref{lem:rootsinfty}, on the rows).

If $l < d$, then we take $\hat v=v$. From \cite[Thm.\,6.3-11, p.\,385]{Kailath80}, by row reducedness, we know that 
$S^{-1}\hat v$ is strictly proper.
If $l\geq d$, the polynomial division can be perfomed by reformulating \cite[Sec.\,9.1, Eq. (2)]{GaGe99} on matrices. 
Since $S$ has non-singular leading matrix, by the predictable degree property \cite[Thm.\,6.3-13, p.\,387]{Kailath80}
we know that the quotient vector~$w$ has degree $\deg v -d$, hence at most $l-d$. Using reversals of matrix polynomials, \cref{eq:divvec} can be rewritten as 
$$
\rev _l(v) = \rev _d(S) \cdot \rev _{l-d}(w) +  x^{l-d+1} \rev _{d-1}(\hat v),
$$
hence we have 
\begin{equation}\label{eq:computrev}
\rev _{l-d}(w) \equiv \rev _d(S)^{-1} \rev _l(v) \bmod x^{l-d+1}. 
\end{equation}
Remark that by reducedness assumption the coefficient matrix of degree $0$ of $\rev _d(S)$ is non-singular, thus the latter matrix 
in invertible modulo $x^{l-d+1}$.
As soon as $w'=\rev _{l-d}(w)$ hence $w=\rev _{l-d}(w')$ are known, then $\hat v$ can be deduced using $\hat v=v-Sw$.
We know that $S^{-1}\hat v$ is strictly proper using reducedness, as done previously. 
Using fast structured matrix arithmetic (\cref{sec:structmat}), $\rev _{l-d}(w)$ is computed from \cref{eq:computrev} and 
$\hat v$ is obtained within the claim cost bound.  

When $d=d_1>d_2$, first we balance the columns degrees.  
With $\delta=d_1-d_2>0$, 
take $D={\text{\rm diag}}(x^{\delta}, \ldots, x^{\delta}, 1, \ldots , 1)$, with $e_1$ entries $x^\delta$.
The matrix $T=SD^{-1}$ has all its column degrees equal to $d_2$. 
Here and below the degree of a rational function 
is the difference between the degrees of the numerator and the denominator. Column and row reducedness are extended accordingly. 

If $l<d_2$ we let $v'=v$, otherwise we can compute a polynomial 
vector $w'$ of degree at most $l-d_2$  and $v'=v - T w'$ of degree less than $d_2$ such that $T^{-1}v'$ is strictly proper. 
This is done using \cref{eq:divvec} after having multiplied everything by $x^{\delta}$ so as to be reduced to a division with polynomial matrices, in time $\softO{n(d+l)}$. This is similar to 
the $d_1=d_2$ case above since $T$ is column reduced.

Then, taking the quotient of the  first $e_1$ entries of $w'$ by $x^{\delta}$, we write 
 $w'= Dw+z$, where $z$ is of degree less than $\delta$ and such that only its first $e_1$ entries may be non-zero. 
The vector $w$ remains of degree at most $l-d_2$, and we obtain $\hat v$ in time $\softO{n(d+l)}$  as 
$$
\hat v = v - Sw= v- S D^{-1} (w'-z) = v' + S D^{-1}z.
$$
In order to complete the proof we check that $S^{-1}\hat v$ is strictly proper. This vector is $S^{-1}\hat v= S^{-1}v'+ D^{-1}z= D^{-1} T^{-1}v'+ D^{-1}z$.
By construction, $T^{-1}v'$ is strictly proper, it is thus the same for $D^{-1} T^{-1}v'$; $z$ has degree at most $\delta -1$ for 
its first~$e_1$ entries (the other ones are zero), hence $D^{-1}z$ is strictly proper. 
\qed 
\end{proof}


\subsection{Normal form modulo the bivariate ideal} \label{sec:polydiv}

Given a  polynomial $f\in \K[x,y]$ whose $y$-degree is less than the dimension $n_y$ of $\Sy$, 
we can apply \cref{lem:matdiv} to the vector $v_y(f)\in \K[x]^{n_y}$ of the coefficients of $f$. 
\Cref{eq:divvec} becomes  
$$
v_y(f) =  \Sy w + v_y(\hat f)
$$
on vectors, and by definition of the Sylvester matrix we have 
$$
\hat f = f - u a - v b \in f+I
$$
for some $u,v\in \K[x,y]$, with $\hat f$ of $x$-degree less than $d$. We show with \cref{prop:division} that, thanks to the uniqueness of the remainder, this allows us 
to define a normal form modulo $\langle a,b\rangle$. The general $y$-degree case for $f$ is treated 
using a preparatory division by $\Sx$ (whose entries are in $\K[y]$) in order to reduce the degree
in $y$.
The overall construction gives a $\K$-linear map
\begin{equation} \label{eq:defmap}
\begin{array}{rl}
\varphi: \K[x,y] \rightarrow & \K[x,y]_{<(d,n_y)}\\
f \mapsto & \hat f = f \rem I  
\end{array}
\end{equation}
such that $f - \varphi(f) \in I$, and $\varphi(g)=0$ if $g\in I$.
The map $\varphi$ is thus appropriate in order to represent the elements in $\Al$ 
by normal forms.

\begin{example}\label{ex:notproper}
With $\K=\mathbb Q$, consider $a=x^2y+y$ and $b=xy^2+x$; we have $d=d_a=2$ and  
$n_y=e_a+e_b=1+2=3$. If 
$f=x$ then both $f$ and $f-b=-xy^2$ are in $\K[x,y]_{<(2,3)}$, hence the map $\varphi$ 
might not be surjective. 
The division as in \cref{eq:divvec} leads to 
$$
v_y (f)= 
\left[\begin{array}{c}
0 
\\
 0
\\
 x
\end{array}\right]
=
\Sy w + v_y(\hat f)
=
\left[\begin{array}{ccc}
x^{2}+1 & 0 & x  
\\
 0 & x^{2}+1 & 0 
\\
 0 & 0 & x  
\end{array}\right]
\left[\begin{array}{c}
0
\\
 0 
\\
 1 
\end{array}\right]
+
\left[\begin{array}{c}
-x  
\\
 0
\\
 0 
\end{array}\right],
$$
and $\hat f=f-b$ since we can check that $S^{-1} v_y(\hat f)$ is strictly proper, 
whereas $S^{-1} v_y(f)$ is not. It may be noted that the quotient algebra 
$\K[x,y]/\langle a,b\rangle$ has dimension $5$, which is smaller than the dimension of 
$\K[x,y]_{<(2,3)}$.
\qed
\end{example}

If $\eta = \deg_y f \geq n_y$ and $\delta = \deg_x f$ is less than the dimension $n_x$ of 
$\Sx$, then we can directly proceed to the division using $\Sx$.   
Otherwise, as we now see with \cref{lem:phi0}, we first extend $\Sx$ to a bigger  appropriate Sylvester matrix 
$\Tx$ of dimension $\delta+1$. 
By linearization, we associate to $f$ a vector $v_x(f) \in \K[y]^{\delta+1}$ of $y$-degree $\eta$. Then using division by $\Tx$, 
whose $y$-degree is the degree $e$ of $\Sx$,
we can compute $f'$ of $y$-degree less than $e < n_y$, such that $f-f'\in I$.

\begin{lemma} \label{lem:phi0}
Assume that the Sylvester matrix $\Sx$ associated to $a$ and $b$ with respect to $x$ is column reduced. 
Consider $f\in \K[x,y]$ of $x$-degree at most $\delta$ and $y$-degree at most $\eta \geq n_y$. Using $\softO{(n_x+\delta) \eta}$ arithmetic operations in $\K$ we can 
compute a polynomial $f'\in \K[x,y]$, of $x$-degree at most $\max\{n_x-1,\delta\}$
and $y$-degree less than $e<n_y$, such that~$f-f'\in I$. 

\end{lemma}
\begin{proof}

If $\delta$ is less than $n_x$, we simply take $\Tx=\Sx$. Otherwise, let $m=\delta-n_x +1$, 
and denote the $y$-leading coefficients of $a,b$ by $s,t \in \K[x]$.
Since $\Sx$ is column reduced, from \cref{lem:rootsinfty} we know that $\gcd(s,t)=1$.
Either $s$ or $t$ is not divisible by $x$, let us assume that it is $s$, and take for 
$\Tx$ over $\K[y]$, the Sylvester matrix associated to $a$ and $x^mb$ 
with respect to $x$. The Sylvester matrix associated to $s$ and $x^mt$ is non-singular, hence 
$\Tx$ is column reduced by \cref{lem:rootsinfty} again:
if either $\deg s=d_a$ or $\deg t=d_b$, then either $\deg s=d_a$ or $\deg t + m=d_b+m$.

This matrix $\Tx$ has dimension $\max\{n_x,\delta+1\}$, and degree $e=\max\{e_a,e_b\}$ in $y$.
The remainder of the division of $v_x(f)$ by $\Tx$ gives $f'$ such that $f-f'\in I$, its $y$-degree is less than the one of  
$\Tx$, and its $x$-degree is less than the dimension of~$\Tx$.
The cost bound is from \cref{lem:matdiv}, with a matrix of dimension $n=\max\{n_x,\delta+1\}$ and degree $e$, and a 
vector of degree $l=\eta\geq e$.
\qed
\end{proof}

\Cref{lem:phi0} allows to first reduce the $y$-degree, the reduction of the degree in $x$ now also ensures the normal form.

\begin{proposition}\label{prop:division}
Assume that the Sylvester matrices $\Sx$ and $\Sy$ associated to $a$ and $b$ are column reduced,
and consider~$f\in \K[x,y]$.  The $\K$-linear map 
in \cref{eq:defmap} is well defined by choosing for~$\hat f$ the unique polynomial in 
$\K[x,y]_{<(d,n_y)}$ such that $f-\hat f\in I$, and $\Sy^{\!\!-1}v_y(\hat f)$ is strictly proper. 
If $f$ has $x$-degree at most $\delta$ and $y$-degree at most~$\eta$, then this 
normal form for $f+I$ in $\Al$ can be computed
using $\softO{(d+\delta)(e+ \eta)}$ arithmetic operations in $\K$.
\end{proposition}
\begin{proof}
We show the existence of such an $\hat f$ for every $f$, then show that $\hat g =0$ if $g\in I$.
After division by $\Sx$ using \cref{lem:phi0} we have $f'\in \K[x,y]$ of $y$-degree less than $e<n_y=e_a+e_b$ such that 
$f-f' \in I$. 
Then by \cref{lem:matdiv}, that is by division by $\Sy$, we obtain $\hat f \in \K[x,y]_{<(d,n_y)}$ 
such that $f'-\hat f \in I$, hence $f-\hat f \in I$. 
By construction, $\Sy^{\!\!-1}v_y(\hat f)$ is strictly proper.
For $g\in I$, this first leads to some $g'$ of $y$-degree less than $e<n_y$.
Since $\Sy$ is column reduced, we know from \cref{lem:rootsinfty}  that the 
$y$-leading coefficients of $a$ and $b$ are relatively prime, hence using~\cite[Lem.\,7]{Lazard85}
there exist polynomials $r,s\in \K[x,y]$ such that 
$$
g'-ra-sb=0, ~\deg_y r < e_b, \textrm{~and~} \deg_y s < e_a. 
$$ 
By uniqueness it follows that we must have $\hat g = 0$ because  
this value is appropriate using above identity.

The map $\varphi(f) =\hat f$ is well defined and provides a normal form. For $f_1, f_2$ in the coset $f+I$ we indeed 
have $\varphi(f_1-f_2)=0$ hence $\varphi(f_1) = \varphi(f_2)$ by $\K$-linearity of the divisions.

From  
 \cref{lem:phi0}, the first division by $\Sx$ costs $\softO{(d+\delta)(e+ \eta)}$, where we use 
 that $\Sx$ has dimension $n_x \leq 2d$ and degree $e$. This leads 
 to the next division of a vector of degree at most $\max\{n_x-1,\delta\}$ by $\Sy$, whose dimension is 
 $n_y < 2e$ and degree~$d$. Using \cref{lem:matdiv} this adds $\softO{e(d+ \delta)}$ operations.
\qed
\end{proof}

\begin{example}
We continue with $a=x^2y+y$ and $b=xy^2+x$ as in \cref{ex:notproper}; $\Sx$ and $\Sy$ have dimension $n_x=n_y=3$.
For $f=y^3+x^3y^2+1$, 
we first reduce the $y$-degree using $\Sx$.
Since $\delta = \deg _x f \geq n_x$, we cannot directly use $\Sx$ which is $3\times 3$.
Following the proof of \cref{lem:phi0} we increase the dimension and
consider $\Tx\in \K[y]^{4\times 4}$, the Sylvester matrix  with respect to $x$ associated to $a$ and $xb$. 
The first division is therefore: 
$$
v_x(f)= 
\left[\begin{array}{c}
y^{2} 
\\
 0 
\\
 0 
\\
 y^{3}+1 
\end{array}\right]
=
\Tx w_1 +v_x(f') =  
\left[\begin{array}{cccc}
y  & 0 & y^{2}+1 & 0 
\\
 0 & y  & 0 & y^{2}+1 
\\
 y  & 0 & 0 & 0 
\\
 0 & y  & 0 & 0 
\end{array}\right]
\left[\begin{array}{c}
0 
\\
 y^{2} 
\\
 1 
\\
 -y  
\end{array}\right]
+
\left[\begin{array}{c}
-1 
\\
 y  
\\
 0 
\\
 1 
\end{array}\right].
$$
The new polynomial is $f'=-x^3+yx^2+1$, its $y$-degree is $1< n_y$, so the division by $\Sy \in \K[x]^{3\times 3}$ in order to reduce the $x$-degree is now possible:
$$
v_y (f')= 
\left[\begin{array}{c}
0 
\\
 x^{2} 
\\
 -x^{3}+1 
\end{array}\right]
=
\Sy w_2 + v_y(\hat f)
=
\left[\begin{array}{ccc}
x^{2}+1 & 0 & x  
\\
 0 & x^{2}+1 & 0 
\\
 0 & 0 & x  
\end{array}\right]
\left[\begin{array}{c}
x  
\\
 1 
\\
 -x^{2} 
\end{array}\right]
+
\left[\begin{array}{c}
-x  
\\
 -1 
\\
 1 
\end{array}\right]
$$
and we obtain the normal form $\hat f=-xy^2-y+1 \in \K[x,y]_{<(2,3)}$. 
\qed
\end{example}

The assumptions of \cref{prop:division} are central to be able to reduce the degree 
in $x$ and also ensure the normal form. The following example describes a situation with the existence 
of  roots at infinity with respect to $y$.

\begin{example} \label{ex:notlemlaz7}
With $\K = \mathbb F_7$, take $a=(x +3) y +x^{2}+5 x +5$ and $b=(x +3) (x +4) y +x^{2}+4 x +2$. Then, the minimal polynomial of $x$ in the quotient algebra is $x+2$ but cannot be obtained by combinations of $a$ and $b$ of $y$-degree less than $e_a=e_b=1$, that is using combinations of the columns of $\Sy$. The vector $\trsp{[0~~x+2]}$ is its own remainder of the division by $\Sy \in \K[x]^{2\times 2}$, hence $x+2$ is not reduced to zero while being in the ideal. 
In this case however, thanks to the random conditioning of \cref{sec:shifts}, we correctly compute the resultant (see \cref{sec:resultant}).
\end{example}

Since the multiplication in $\K[x,y]$ can be computed in quasi-linear time~\cite[Sec.\,8.4]{GaGe99}, 
\cref{prop:division} allows  multiplication in $\K[x,y]/\langle a,b\rangle$ using 
$\softO{de}$ arithmetic operations. This is valid as soon as both Sylvester matrices are column reduced. 
From \cref{lem:rootsinfty} this means that the $x$-leading (resp. $y$-) coefficients  of $a$ and $b$ 
are coprime and one of them has maximal degree $d_a$ or $d_b$ (resp. $e_a$ or $e_b$).
In a complementary situation, that is with a sufficiently generic ideal $\langle a , b\rangle$ 
for the graded lexicographic order 
and using the total degree, a quasi-linear complexity was already achieved in \cite{HoevenLarrieu2019} for the multiplication in such a quotient. 
Even though it retains specific assumptions on the ideal, 
let us also mention the multiplication bound $\softO{(de)^{1.5}}$ of \cite[Sec.\,4.5]{HMSS19}.


\subsection{Power projections via transposed normal form} \label{sec:Tdivision}

Using Shoup's general approach for the computation of minimal polynomials in a quotient algebra, we especially rely on the fact that the power projection problem is the transpose of the modular composition
problem~[\citealp{Shoup94}; \citealp[Sec.\,6]{Kal00-2}]. The normal form algorithm of \cref{prop:division} treats a special case of 
modular composition since $f \bmod I$ can be seen as $f(g(x)) \bmod I$ for $g=x$. 
Certain power projections can therefore already be derived by transposition from what we have done so far, as we explain in this section.  
This is used at the core of the general algorithm in \cref{sec:composition} for the computation 
of a larger number 
of $O(de)$ projections efficiently for $\K=\mathbb F_q$.

Consider 
the restriction $\varphi_{\delta,\eta}$  of $\varphi$ to the $\K$-vector 
space~${\mathcal U}=\K[x,y]_{\leq(\delta,\eta)}$, and denote $\K[x,y]_{<(d,n_y)}$ as a 
$\K$-vector space by $\mathcal V$.
We also introduce the dual spaces $\widehat {\mathcal U}$ and $\widehat{ \mathcal V}$ of the $\K$-linear forms on ${\mathcal U}$ and ${\mathcal V}$, respectively. 
The transpose of $\varphi_{\delta,\eta}$ is the $\K$-linear map 
\begin{equation}   \label{eq:linearform}
\begin{array}{rl}
\trsp{\varphi}_{\delta,\eta}: {\widehat {\mathcal V}} \rightarrow & {\widehat {\mathcal U}}\\
\ell \mapsto &  \ell \circ {\varphi}_{\delta,\eta}.  
\end{array}
\end{equation}
We view the polynomials in ${\mathcal U}$ as vectors on the monomial basis 
${\mathcal B}=\{1, x, \ldots, x^{\delta}, y, xy, \ldots, x^{\delta}y^{\eta}\}$. 
The linear forms in~${\widehat {\mathcal U}}$ on the dual basis 
of ${\mathcal B}$ are represented by vectors in $\K^{(\delta+1)(\eta+1)}$. 
The elements in ${{\mathcal V}}$ and ${\widehat {\mathcal V}}$ are viewed in $\K^{dn_y}$ on  
the basis $\{y^{n_y-1}, y^{n_y-1} x, \ldots, y^{n_y-1} x^{d-1}, 
y^{n_y-2}, y^{n_y-2} x, \ldots, x^{d-1}\}$ of 
${{\mathcal V}}$ (in accordance with the definition of the Sylvester matrix $\Sy$). 
From \cref{eq:linearform}, the entries of $\trsp{\varphi}_{\delta,\eta} (\ell)$  are the bivariate power projections
\begin{equation} \label{eq:powerproj}
(\ell \circ \varphi_{\delta,\eta})({x^{i}y^{j})} \textrm{~~for~} 
0\leq i \leq \delta \textrm{~and~}  0\leq j \leq \eta.
\end{equation}

We compute these projections by applying the transposition principle~[\citealp{Borde56,Fidu73}; \citealp{Shoup94}]. 
The principle asserts that if a $\K$-linear map $\phi: {\mathcal E}_1 \rightarrow {\mathcal E}_2$ can be computed by a linear straight-line program of length $l$, then the transpose map can be computed 
by a program of length $l+ \dim {\mathcal E}_2$ ($l$ if $\phi$ is an isomorphism)~\cite[Thm.\,13.20]{BCS97}. 
We use a commonly applied strategy to implement the principle~\cite{BLS03}. \cref{prop:trspred} follows directly from \cref{prop:division} e.g. by mimicking the results of
\cite[Sec.\,4]{BFSS06} in $\K[x,y]/\langle f(x),g(y)\rangle$ or \cite{PaSc06,PS13} modulo triangular sets.  The change concerns only the way in which 
the ideal is represented. 

\begin{proposition} \label{prop:trspred}
Assume that the Sylvester matrices $\Sx$ and $\Sy$ associated to $a$ and $b$ are column reduced.
Given two integers $\delta,\eta\geq 0$ and $\ell \in {\widehat {\mathcal V}}$ one can compute 
 $\trsp{\varphi}_{\delta,\eta} (\ell)$
using $\softO{(d+\delta)(e+ \eta)}$ arithmetic operations in $\K$. 
\end{proposition}
\begin{proof}
The claim on $\trsp{\varphi}_{\delta,\eta}$ is going to follow from the application of the transposition principle to the algorithm of \cref{prop:division} for ${\varphi}_{\delta,\eta}$, 
with portions written as a $\K$-linear straight-line program.

The computation of ${\varphi}_{\delta,\eta}$ reduces to two applications of \cref{lem:matdiv}, hence 
it suffices to study the transposition of the matrix division with remainder algorithm. 
Regarding this algorithm, observe that if the matrix inverses such as in \cref{eq:computrev} are pre-computed, then afterwards only $\K$-linear forms in the entries of the input vector
 are involved. Furthermore, 
these linear forms can be computed by $\K$-linear straight-line computations. 
The division algorithm of \cref{lem:matdiv} can therefore be viewed as follows. 
The inverse of the reversed Sylvester matrix  in \cref{eq:computrev} is pre-computed 
over truncated power series in time as stated in \cref{lem:matdiv}, using the complexity bounds in \cref{sec:structmat}.  
Then the linear operations involving the input vector, including structured matrix times vector products~\cite{BiPa94}, are performed within the same cost bound.  
The transposed division algorithm follows: it uses the pre-computed inverse as a parameter, and it is obtained from 
the transposition principle applied to the linear straight-line remaining portions. 
For the transposed division, this leads to the same 
complexity bound as stated in \cref{lem:matdiv}, and  for the transpose $\trsp{\varphi}_{\delta,\eta}$,
to the bound as in \cref{prop:division}. 
\qed 
\end{proof}

We directly use the transposition principle. The transpose algorithm could nevertheless be stated more explicitly as done for the univariate case in~\cite{BLS03}\,---\,using duality with linear recurrence sequence extension~\cite{shoup91}, and for multivariate triangular sets in~\cite{PaSc06,PS13}.

%
%

\section{Application of Kedlaya and Umans' techniques} \label{sec:composition}

The minimal polynomial of the multiplication by $x$ in $\Al$ requires the computation of projections of 
$O(de)$ power of~$x$~(\cref{sec:invfact}). \cref{prop:trspred} therefore is not sufficient in order 
to achieve quasi-linear complexity. This now leads us to apply Kedlaya and Umans' techniques~\cite{KU11},
and their extensions in \cite{PS13,HoeLec21a}, for efficient modular composition over a finite field 
(\cref{cor:divisionall})
and power projection by transposition (\cref{cor:powerproj}).

Given
three polynomials $f,g,h \in \K[x]$ with 
$\deg(f) < n$ and $\deg(g) < n$ where \(n=\deg(h)\), 
the problem of modular composition is to compute $f(g) \bmod  h$~\cite{BK78}. 
(The problem is more fundamentally stated over a ring.)
At the very beginning in this case, we benefit from the fact that  for such polynomials 
the division with remainder can be computed using $\softO{n}$ arithmetic operations~\cite[Sec.\,9.1]{GaGe99}. 
One of the difficulties in the bivariate case is to be able to start from an analogous point, we mean from an efficient division with remainder modulo $I$. Once this is achieved, 
the approach of~\cite{KU11} can be followed for both modular composition and power projection. 
This is what has been accomplished in~\cite{PS13} (multivariate case) 
and \cite{HoeLec21a} (special case $g=x$), with respective shapes of the ideal $I$ that we have already mentioned.  
We proceed in the same way, and integrate the new division (normal form) algorithm into the overall process.  
We therefore do not repeat all the details for the proof of \cref{thm:KU} and its corollaries, and refer the reader to the stem papers.  
As for \cref{prop:trspred} our change is the way in which 
the ideal is represented, which leads to a new modular bivariate projection algorithm in \cref{cor:powerproj}.  

The first main ingredient is to reduce the problem of division (of modular composition), to divisions 
with smaller input degrees and to a problem of multipoint evaluation~\cite[Pb.\,2.1]{KU11}. More precisely, 
\cref{thm:KU} shows that the problem of computing the normal form of $f \in \K[x]_{<\delta}$ modulo $I$ can be reduced, for $2 \leq d_{\epsilon} < \delta$, to normal forms of polynomials of $x$- and $y$-degrees less than $d_{\epsilon}d\log \delta$ 
and $d_{\epsilon}e\log \delta$, respectively, and to multipoint evaluation. Here, remember the notations
 $d=\max\{deg_x \,a, deg_x \,b\}$ and $e=\max\{deg_y \,a , deg_y \,b\}$. The Sylvester matrix $\Sy$ is
  $n_y\times n_y$ over$\K$.

\begin{theorem}[{\normalfont \cite[Thm.\,3.1]{KU11}, generalized in \cite{PS13,HoeLec21a}}] 
\label{thm:KU}
Consider 
$f\in \K[x]$ of degree less than~$\delta$, and an arbitrary integer $2 \leq d_{\epsilon} < \delta$. 
Assume that the Sylvester matrices $\Sx$ and $\Sy$ associated to $a$ and $b$ are column reduced, 
and $|\K|> l (d_{\epsilon} -1)\max\{d-1,n_y-1\}$ where 
$l=\lceil \log_{d_{\epsilon}}(\delta)\rceil$. 
If $\delta=O(de)$ then  $f(x) \rem I $ can be computed using 
$\softO{d_{\epsilon}^2 de}$ arithmetic operations in $\K$,  plus one multivariate multipoint evaluation 
of a polynomial with $l$ variables over $\K$ and individual degrees less than $d_{\epsilon}$, at 
$O(l^2d_{\epsilon}^2 de)$ points in~$\K^l$. 
\end{theorem}

\begin{proof}
The following six steps are those of the proof of \cite[Thm.\,3.1]{KU11}.

\begin{enumerate}
   \item \label{step:S1} We first appeal to the inverse Kronecker substitution \cite[Dfn.\,2.3]{KU11}, in order to map 
$f$ to a polynomial with $l$ variables and degree less than $d_{\epsilon}$ in each variable. 
This $\K$-linear map from $\K[x]_{< \delta}$ to  
$\K[z_0,\ldots, z_{l-1}]_{<(d_{\epsilon}, \ldots, d_{\epsilon})}$ is defined as follow.
For $0\leq k < \delta$,  the monomial
$x^k$ is sent to $z_0^{k_0}z_1^{k_1}\ldots z_{l-1}^{k_{l-1}}$, where $k_0,k_1,\ldots, k_{l-1}$ are the 
coefficients of the expansion of~$k$ in base $d_{\epsilon}$. This is extended linearly to $\K[x]_{< \delta}$, 
and $f$ is mapped in this way to a polynomial 
$\phi \in \K[z_0,\ldots, z_{l-1}]_{<(d_{\epsilon}, \ldots, d_{\epsilon})}$. 
The map is injective on $\K[x]_{< \delta}$ and is computed in linear time using monomial bases.

   \item Then we compute the polynomials $\chi _i = x^{d_{\epsilon}^{i}} \rem I$ in $\K[x,y]_{<(d,n_y)}$, for 
$i=0,\ldots, l-1$. This corresponds to~$l$ exponentiations  by $d_{\epsilon}$ modulo~$I$. 
By successive bivariate multiplications \cite[Sec.\,8.4]{GaGe99}, each followed by a reduction
modulo the ideal, this can be done in time $\softO{de}$ from \cref{prop:division}.

 A key property is that 
 $f(x) \rem I = \phi(\chi _0, \ldots, \chi _{l-1}) \rem I$. This leads to the idea of first computing
   $\phi(\chi _0, \ldots, \chi _{l-1})$ by evaluation-interpolation, and to perform only afterwards the reduction modulo the ideal. 
We have that the degree of $\phi(\chi _0, \ldots, \chi _{l-1})$ is at most $\delta' = l (d_{\epsilon} -1)(d-1)$ in $x$, and 
$\eta' = l (d_{\epsilon} -1)(n_y-1)$ in $y$. 

 \item We choose subsets $K_1$ and $K_2$ of $\K$ or cardinalities $\delta'+1$ and $\eta'+1$, respectively.
By multipoint bivariate evaluation,  
we compute all values $\mu _{i,j,k}= \chi _i(\lambda_j,\lambda_k) \in \K$ 
for $i=0,\ldots, l-1$  and $(\lambda_j,\lambda_k)\in K_1\times K_2$.
Using univariate evaluation~\cite[Sec.\,10.1]{GaGe99}, variable by variable, this can be done using 
$\softO{l \delta' \eta'}$ hence $\softO{d_{\epsilon}^2de}$ operations ($n_y\leq 2e$). 

 \item \label{step:S4} This is followed by all the evaluations $\phi(\mu _{1,j,k}, \ldots, \mu _{l,j,k})$, 
which is multipoint evaluation of a polynomial with $l$ variables, with individual degrees 
less than  $d_{\epsilon}$, at $(\delta'+1)(\eta'+1)$ points in $\K^l$.

 \item \label{step:S5} From there, $\phi(\chi _0, \ldots, \chi _{l-1})$ is recovered using bivariate interpolation from its values just obtained at  $K_1\times K_2$, this uses $\softO{d_{\epsilon}^2de}$ operations in a way similar to 
multipoint bivariate evaluation above. 

 \item \label{step:S6} Finally, $f(x) \rem I = \phi(\chi _0, \ldots, \chi _{l-1}) \rem I$. We know from 
\cref{prop:division} that this costs $\softO{\delta'\eta'}$ operations, which is  $\softO{d_{\epsilon}^2de}$.
\qed 
\end{enumerate}

\end{proof}

In line with \cite[Thm\,7.1]{KU11} and \cite{PS13,HoeLec21a}, thanks to fast multipoint evaluation~\cite[Cor.\,4.5]{KU11}, we now can bound the cost of the reduction of a univariate polynomial modulo the ideal.

\begin{corollary} \label{cor:divisionall}
Let $\K$ be a finite field $\mathbb F_q$. 
Assume that the Sylvester matrices $\Sx$ and $\Sy$ associated to $a$ and $b$ are column reduced, and 
consider $f\in \mathbb F_q[x]$ of degree less than $\delta=4de$. For every constant $\epsilon>0$,
if $q\geq \delta^{1+\epsilon}$, 
then~$f \rem I$ can be computed using $O((de)^{1+\epsilon}\log(q) ^{1+o(1)})$ bit operations. 
\end{corollary}
\begin{proof}
Depending on $\epsilon$, we choose a large enough constant integer $c$ 
for $d_{\epsilon}=\lceil\delta ^{1/c}\rceil$ to be sufficiently small compared to~$de$. We have in particular, 
$d_{\epsilon}<\delta^{\epsilon}$, and $l=\lceil \log_{d_{\epsilon}}(\delta)\rceil\leq c$.
For $\delta$ large enough this leads to  
$l (d_{\epsilon} -1)\max\{d-1,n_y-1\} \leq q$ and therefore we can apply \cref{thm:KU}. 
We know that  $f \rem I$ can be computed using $\softO{d_{\epsilon}^2 de}$ operations in~$\mathbb F_q$, 
which is $O(de)^{1+\epsilon}$, plus the cost of the associated multipoint evaluation.  
Then we use the fact that for every constant $\gamma >0$, there is an algorithm for evaluating a polynomial in 
$\mathbb F_q[z_0,\ldots, z_{l-1}]_{<(d_{\epsilon}, \ldots, d_{\epsilon})}$
at $n$ points in~${\mathbb F}_q^l$ using $(d_{\epsilon}^l+n)^{1+\gamma}\log(q)^{1+o(1)}$
bit operations, when the individual degrees $d_{\epsilon}$ are sufficiently large, and the number of variables~$l$ is at most $d_{\epsilon}^{o(1)}$~\cite[Cor.\,4.5]{KU11}.
Considering the evaluation parameters in \cref{thm:KU} and $l\leq c$, for every$\gamma>0$, 
the cost of the multipoint evaluation then is $O((\delta +d_{\epsilon}^2 de)^{1+\gamma}\log(q)^{1+o(1)})$,
which allows to obtain the claimed complexity bound. 
\qed
\end{proof}

As it has been said before, our presentation is simplified compared to the one of \cite{KU11}. The dependence in $q$,
in complexity bounds analogous to the one in \cref{cor:divisionall}, 
is made explicit and written  using polylogarithmic functions in \cite{PS13}. The study of 
 \cite{HoeLec21a} uses an explicit function of slow increase for the number of variables of the 
 multipoint evaluation problem.  
Sharper bounds and improved algorithms  can be found in \cite{HoeLec20c,BGGKU22} for  multipoint evaluation,
and~\cite{HoeLec21d} for multivariate modular composition over finite fields. Since we rely also on a solution for the dual problem, and that it is not treated  
in the latter references, we remain essentially based on \cite{KU11}.

In a similar way to what we did in \cref{sec:Tdivision} we now transpose the algorithm of 
\cref{cor:divisionall}.
Our reasoning is that of \cite[Thm\,7.7]{KU11}, \cite[Thm.\,3.3]{PS13} and \cite[Prop.\,1]{HoeLec21a} for modular power projection. We use the notation $\varphi$ of \cref{eq:defmap} for the normal form map.

\begin{corollary} \label{cor:powerproj}
Let $\K$ be a finite field $\mathbb F_q$. 
Assume that the Sylvester matrices $\Sx$ and $\Sy$ associated to $a$ and $b$ are column reduced. 
Let $\ell$ be a linear form in the dual of $\mathbb F_q[x,y]_{<(d,n_y)}$.
For every constant $\epsilon>0$, 
if $q\geq \delta^{1+\epsilon}$ with $\delta=4de$, 
then  the projections 
$
(\ell \circ \varphi)({x^{i})} 
$
for $0\leq i < \delta$  can be computed using $O((de)^{1+\epsilon}\log(q) ^{1+o(1)})$ bit operations. 
\end{corollary}
\begin{proof}

From \cref{eq:powerproj}, we have to compute $\trsp{\varphi}_{\delta-1,0} (\ell)$. 
The claim follows from the transposition principle (\cref{sec:Tdivision}) applied to the successive 
algebraic steps of
the normal form algorithm of \cref{cor:divisionall}, in reverse order. 
The non-algebraic portions of the algorithm involved in multipoint evaluation are treated 
by means of~\cite[Thm.\,7.6]{KU11}. 
The steps of the algorithm are given in the proof of \cref{thm:KU}~\cite[Thm.\,3.1]{KU11}. The four of them that depend on the input $f$ have to be considered, these are 
Steps~\ref{step:S1}, \ref{step:S4}, \ref{step:S5}, and \ref{step:S6}, that we see 
as $\mathbb F_q$-linear maps.
The last step~\ref{step:S6} is reduction modulo~$I$, the transpose is obtained from \cref{prop:trspred}.
Step~\ref{step:S5} is bivariate interpolation, computed by interpolating in~$x$ then in~$y$.  This is transposed 
using two transposed univariate interpolation~\cite{KL88,BLS03}.
Step~\ref{step:S4} is multivariate evaluation using~\cite[Cor.\,4.5]{KU11}. The transpose is given by~\cite[Thm.\,7.6]{KU11} 
when the ambient dimension is equal to the number of evaluation points, i.e. the linear map can be represented by a square matrix. The general case in which we are, with a larger number of evaluation points, is treated as in the proof of~\cite[Thm.\,7.7]{KU11} using several instances of the square case with a cost that fits the claimed bound. Finally, the transpose of the 
inverse Kronecker substitution at Step~\ref{step:S1} is a projection that takes linear time. 
In view of the transposition principle and of ~\cite[Thm.\,7.6]{KU11}, the algorithm obtained from those transpositions 
computes the power projections using $O((de)^{1+\epsilon}\log(q) ^{1+o(1)})$ bit operations, as in \cref{cor:divisionall}.
\qed
\end{proof}

%
%

\section{Non-singular leading matrices using random shifts and reversals} \label{sec:shifts}

In order to exploit the powers projections of \cref{cor:powerproj} for a minimal polynomial computation, 
and derive the last invariant factor of the Smith normal form of $\Sy$, we need to address a column reducedness issue. This is what we do in this section. If the input polynomials $a$ and $b$ lead to 
 $\Sx$ and $\Sy$ with singular leading coefficients, then we construct two new polynomials $a'$ and $b'$
 which allow to get around the difficulty.

\begin{lemma}[Conditioning of \Sx] \label{lem:condSx}
Given $\alpha\in \K$ not a root of $\res_x(a,b)$ (the ideal is zero-dimensional), in arithmetic time $\softO{de}$ we can compute 
two polynomials $a'$ and $b'$ with degrees as those of $a$ and $b$, such that the new Sylvester 
matrix $\Sxp$ is column reduced and the Smith normal form of $\Syp$ is that of $\Sy$.
\end{lemma}
\begin{proof}
Consider  $a^{(1)}(x,y)=a(x,y+\alpha)$ and $b^{(1)}(x,y)=b(x,y+\alpha)$.
The new Sylvester matrix  $\Sxk{1}$ with respect to $x$ has a non-singular constant term since $(\res_x(a,b))(\alpha)\neq 0$. 
The Smith normal form of $\Syk{1}$ is equal to the Smith normal form of $\Sy$. Indeed, 
let $Q_{\alpha,k}\in \K^{k\times k}$ be the matrix of the endomorphism that shifts a polynomial 
of degree less than $k$ by $\alpha$; $Q_{\alpha,k}$ is lower triangular with unit diagonal.  We 
have
\begin{equation}\label{eq:Sy1}
\Syk{1} = Q_{\alpha,e_a+e_b} \, \Sy \, {\text{\rm diag}}(Q_{\alpha,e_b}^{-1},Q_{\alpha,e_a}^{-1}), 
\end{equation}
hence $\Syk{1}$ and $\Sy$ are unimodularly equivalent. 
 Then we consider the reversed polynomials $a'$ and $b'$ of $a^{(1)}$ and $b^{(1)}$
 with respect to $y$, using the respective degrees $e_a$ and $e_b$. 
 Note that $a'$ and $b'$ must keep the same $y$-degrees, otherwise~$\Sxk{1}$ could not have a non-singular constant coefficient; for the same reason, the new matrix $\Sxp$ associated to  $a'$ and $b'$ is column reduced.
On the other hand, the Smith form with respect to $y$ is unchanged since 
\begin{equation}\label{eq:Sy2}
\Syp = J_{e_a+e_b} \, \Syk{1} \, {\text{\rm diag}}(J_{e_b},J_{e_a}),
\end{equation}
where $J_k$ is the reversal matrix of dimension $k$. The cost is dominated by the one of at most $2(d+1)$ 
shifts of polynomials of degree at most~$e$ in $\K[y]$, see e.g. \cite{BiPa94}[Chap. 1, Pb. 3.5]. 
\qed
\end{proof}

Note that \cref{lem:condSx} preserves the Smith normal form of $\Sy$ but not necessarily its Hermite form.  From \cref{lem:rootsinfty}, $a'=b'=0$ has no roots at infinity with respect to $y$,  
so the last invariant factor of $\Sy$ is the minimal polynomial of the multiplication by $x$ in the new quotient algebra (\cref{lem:laz}). 
The latter may have changed, with an extra factor coming from 
possible roots at infinity for $a=b=0$.  

We now do the same type of manipulation for the column reducedness of $\Sy$ and need a preliminary 
observation on reversed polynomial matrices. The reversal by columns of a matrix polynomial is 
the matrix whose entries are reversed with respect to the degree of their column.

\begin{lemma} [Reversed Smith normal form] \label{lem:smithrev}
The last invariant factor of the reversal of $A\in \K[x]^{n\times n}$ by columns
is the reversal of the last invariant factor of $A$ made monic and  multiplied by some power of $x$. 
\end{lemma}
\begin{proof}
Let $X$ in $\K[x]^{n\times n}$ with a determinant which is a power of $x$ be such that the reversal $R$
of $A$ by columns is~$A(1/x)X$. 
Let $\SA$ be the Smith normal form of $A$, with unimodular matrices $U$ and $V$ such that $AV=U\SA$.
 We have 
\begin{equation}\label{eq:tmprev}
R X^{-1} V(1/x) = U(1/x)\SA(1/x).
\end{equation}
Let $\SAe$ be the diagonal matrix whose entries are the reversals of the diagonal entries of $S$, made monic by division by their leading coefficients. By multiplying \cref{eq:tmprev}  by an appropriate power of $x$,
we obtain
$$
R W_1 = W_2 \SAe
$$
for two  matrices $W_1$ and $W_2$ in $\K[x]^{n\times n}$ whose determinants are powers of $x$. 
Now let  
$\SRe$ be the diagonal matrix whose diagonal entries are the invariant factors of $R$ divided 
by the largest power of $x$ they contain. Using similar manipulations as above we get 
$$
\SRe W_3 = W_4 \SAe
$$
for two  matrices $W_3$ and $W_4$ in $\K[x]^{n\times n}$ whose determinants are also powers of $x$.
By the multiplicativity of the Smith normal form~\cite[Ch. II, Thm.\,2.15]{New72}, noting that 
$S^*_A$ and $S^*_R$ are themselves in Smith normal form, we arrive at $\SRe=\SAe$.
The claim follows since the last invariant factor of $R$ is the one of $\SRe$ multiplied by some 
power of $x$, and the last invariant factor of $\SAe$ is the reversal of the last invariant factor of $A$ divided by its leading coefficient.
\qed
\end{proof}

\begin{lemma}[Conditioning of \Sx \, and \Sy] \label{lem:SxSy}
Given $\alpha,\beta\in \K$ not  roots of $\res_x(a,b)$ and $\res_y(a,b)\in \K[x]$, respectively (the ideal is zero-dimensional),  in arithmetic time $\softO{de}$ we can compute 
two polynomials $a'$ and $b'$ with degrees as those of $a$ and $b$ such that: the new Sylvester matrices  $\Sxp$ and $\Syp$ are column reduced;  the last invariant factor of 
$\Sy$ can be deduced from that of $\Syp$ using  $\softO{de}$ additional arithmetic operations. 
\end{lemma}
\begin{proof}
By applying \cref{lem:condSx} we can assume that $a$ and $b$ are such that $\Sx$ is column reduced,
without modifying~$\Sy$ so $\res_y(a,b)$ either.
We use arguments similar to those used in the proof of \cref{lem:condSx}.
We first take $a^{(1)}(x,y)=a(x+\beta,y)$ and $b^{(1)}(x,y)=b(x+\beta,y)$.
The new Sylvester matrix  $\Syk{1}$ with respect to $y$ has a non-singular constant term since 
$(\res_y(a,b))(\beta)\neq 0$. 
We denote the last invariant factor of $\Sy$ by $\sigma \in \K[x]$.
The last invariant factor of $\Syk{1}$ is  $\sigma_\beta=\sigma(x+\beta)$ and satisfies 
 $\sigma_\beta(0)\neq 0$.
Then we consider the reversed polynomials $a'$ and $b'$ of $a^{(1)}$ and $b^{(1)}$
 with respect to $x$, using the respective degrees $d_a$ and $d_b$. 
 Since $\Syk{1}$ has a non-singular constant term, 
 $a'$ and $b'$ keep the same $x$-degrees and the new matrix $\Syp$ associated to  $a'$ and $b'$ is column reduced. 

We now prove the claims with $a'$ and $b'$.
We have just seen for the column reducedness of the $y$-Sylvester matrix. 
With respect to $x$, the Sylvester matrix is column reduced after the initial application of \cref{lem:condSx}.
Using \cref{eq:Sy1,eq:Sy2} from the proof of the latter lemma, now with $\beta$, $\Sxk{1}$, and $\Sxp$, 
we deduce that $\Sxp$ remains column reduced. 
Finally, we compute the last invariant factor $\sigma$ of $\Sy$ from the one of $\Syp$.
The Sylvester matrix $\Syp$ is the reversal of~$\Syk{1}$.  
Let $\sigma_\beta'$ the reversal polynomial of $\sigma_\beta$. From \cref{lem:smithrev} we deduce that  
the last invariant factor of $\Syp$ is $c x^l \sigma_\beta'$ for some integer $l\geq 0$, and 
a non-zero $c \in \K$. Since $\sigma_\beta(0)\neq 0$, the reversal of $c x^l \sigma_\beta'$ is $c \sigma_\beta$. Using  a shift by $-\beta$ and making the polynomial monic provides us with $\sigma$.

In addition to the cost in \cref{lem:condSx}, 
we essentially have to perform at most $2(e+1)$ 
shifts of polynomials of degree at most~$d$ in~$\K[x]$,  
plus a final shift of a polynomial of degree $O(de)$, see e.g. \cite{BiPa94}[Chap. 1, Pb. 3.5].
\qed
\end{proof}

%
%

\section{Invariant factor computation} \label{sec:invfact}

For an appropriate random linear form $\ell$, the minimal polynomial $\mu$ of 
the multiplication by $x$ in $\K[x,y]/\langle a,b\rangle$ is also 
the one of the linearly generated sequence 
$(\ell \circ \varphi)({x^{i})}_{i\geq 0}$ with high probability [\citealp[Sec.\,4]{Shoup95}; \citealp[Lem.\,6]{KaSh98}]. In essence, this minimal polynomial approach is a transcription of that of  Wiedemann~\cite{Wie86},  
with multiplication matrices rather than sparse ones~\cite{Shoup94}.
This allows, in this section, to first bound the complexity of the minimal polynomial problem 
from the power projection complexity bound we have obtained previously (\cref{cor:powerproj}).
Since $\mu$ has degree at most $2de$, it can ideed be computed 
from the first $4de$ terms of the power projection sequence. 
However, this is only valid when the involved Sylvester matrices are column reduced.
Up to random shifts and reversals (\cref{sec:shifts}), we then describe how the last invariant factor of $\Sy$ can be derived from the minimal polynomial of the multiplication by $x$ in a slightly modified quotient algebra.

\begin{theorem} \label{thm:1}
Consider two polynomials $a,b \in \mathbb F_q[x,y]_{\leq (d, e)}$ and 
assume that the associated Sylvester matrices $\Sx$ and~$\Sy$  are column reduced.
For every constant $\epsilon>0$, 
if {$q\geq \delta^{1+\epsilon}$} with $\delta=4de$,  
there exists a randomized Monte Carlo algorithm which computes the 
minimal polynomial of the multiplication by $x$ in $\mathbb F_q[x,y]/\langle a,b\rangle$
using  $O((de)^{1+\epsilon}\log(q) ^{1+o(1)})$ bit operations. The algorithm returns 
a divisor of the minimal polynomial, to which it is equal with probability at least~$1-2de/q\geq1/2$.
\end{theorem}
\begin{proof}
The modular power projections as in \cref{cor:powerproj} are computed for a random linear map $\ell$. 
The sequence $\{(\ell \circ \varphi)({x^{i}})\}_{i\geq 0}$ is linearly generated; its minimal polynomial $\mu'$ is a divisor of the minimal polynomial $\mu$ of the multiplication by $x$ in $\Al$. Since $\deg \mu \leq 2de$, $\mu'$ can be computed using $\softO{de}$ additional operations in $\K$ from the $4de$ first terms of the sequence~\cite[Algo\,12.9]{GaGe99}. 
We can conclude by proving that $\mu'=\mu$ with high probability. Following the construction of $\varphi$ in \cref{eq:defmap}, one can define the multiplication map 
\begin{equation} \label{eq:mulmap}
\begin{array}{rl}
\psi: \mathbb F_q[x,y]_{<(d,n_y)} \rightarrow & \mathbb F_q[x,y]_{<(d,n_y)}\\
f \mapsto & xf \rem I.  
\end{array}
\end{equation}
For an appropriate basis of $\mathbb F_q[x,y]_{<(d,n_y)}$ as a $\mathbb F_q$-vector space, we consider that 
$\psi$ is represented by a matrix $M\in \mathbb F_q^{(dn_y) \times (dn_y)}$ and that $1$ is represented by the vector $u \in \mathbb F_q^{dn_y}$.
According to what we have seen in  \cref{sec:Tdivision}, we also represent linear forms in the 
dual of $\mathbb F_q[x,y]_{<(d,n_y)}$ by vectors in $\mathbb F_q^{dn_y}$. 
With this, $\mu$ is the minimal polynomial of  $u$ with respect to $M$. Hence for a random linear form  $\ell$ represented by $v\in \mathbb F_q^{dn_y}$, the minimal polynomial 
of the linearly generated sequence $\{(\ell \circ \varphi)({x^{i}})\}_{i\geq 0}=\{\trsp{v}M^iu\}_{k\geq 0}$ is $\mu$ with probability 
at least $1-\deg\mu/q$ [\citealp[Lem.~2]{KaPa91}; \citealp[Lem.~1]{KaSa91}]. 
\qed
\end{proof}

\begin{corollary} \label{cor:invfact}
Consider two coprime polynomials $a,b \in \mathbb F_q[x,y]_{\leq (d,e)}$.
For every constant $\epsilon>0$,  there exists a randomized Monte Carlo algorithm which computes the last invariant factor of the Sylvester matrix associated to $a$ and~$b$ with respect to either~$x$ or $y$,
using 
 $O((de)^{1+\epsilon}\log(q) ^{1+o(1)})$ bit operations. 
The algorithm either returns the  target invariant factor,  and this with probability at least $1/2$, 
one of its divisors, or ``failure''. 
\end{corollary}
\begin{proof}
When  $q\geq (12de)^{1+\epsilon}$, we randomly choose random $\alpha$ and $\beta$ in $\mathbb F _q$, 
then check whether $\Sxp$ and $\Syp$ as in \cref{lem:SxSy} are column reduced. The check is performed
using $\softO{d+e}$ operations, see \cref{lem:laz} and e.g. \cite[Thm.\,11.10]{GaGe99}. 
Since $\res_y(a,b)\in \mathbb F_q[x]$ and $\res_x(a,b)\in \mathbb F_q[y]$ have degree at most $2de$, 
 the probability of success is at least  $1-4de/q$.
If the Sylvester matrices are column reduced, from \cref{thm:1}, we then compute the minimal polynomial of the multiplication by $x$ (or $y$) in the 
quotient algebra associated to $\Syp$ (or $\Sxp$). 
\Cref{lem:laz} tells us that we have actually computed the last invariant factor of $\Syp$ (or $\Sxp$)
with probability at least $1-2de/q$.
From \cref{lem:SxSy} again, we finally derive the last invariant factor of~$\Sy$~(or $\Sx$).
If $q$ is too small, we construct an extension field of $\mathbb F_q$ with cardinality at least 
$(12de)^{1+\epsilon}$, that is of  degree~$O(\log (de))$. This can be done using an expected number of 
$\softO{(\log (de)^2 + \log (de)\log (q)}$ bit operations~\cite{Shoup94} (see also~\cite{CouveignesLercier2013} and \cite[Sec.\,14.9]{GaGe99} in this regard).  
We then work in this extension,  the costs induced are logarithmic factors which 
do not change our target cost bound, and the probability of success can be adjusted. 
\qed
\end{proof}

%
%

\section{Elimination ideal and resultant} \label{sec:resultant}

When the system $a=b=0$ has no roots at infinity with respect to $y$, from \cref{lem:laz} and \cref{cor:invfact} we obtain a Monte Carlo algorithm for computing 
the minimal polynomial of the multiplication by $x$, it is a {generator of the elimination ideal} $\langle a, b\rangle \cap \mathbb F_q[x]$. 

Still with the absence of roots at infinity with respect to $y$, \cref{lem:shape} indicates that if, moreover, the ideal 
has a shape basis $I=\langle \mu(x), y-\lambda(x)\rangle$~\cite{GTZ88,BMMT94}, then the {resultant} of $a$ and~$b$ is known. 
Note  that the extra non-zero constant in \cref{lem:shape} can be computed
at the cost of $\softO{de}$ operations in $\mathbb F_q$ using evaluation in $x$. 

For the resultant, we see that this leads to a weaker genericity assumption than in \cite{HoeLec21a}, where the total 
degree is used.  Assume that the ideal~$\langle a,b\rangle$ is in generic position for the lexicographic order 
$y>x$ so that   $res_y(a,b)=c \mu$ with~$c\neq 0\in \mathbb F_q$. 
In this case, from \cref{lem:shape} and \cref{lem:laz} again, we can compute the resultant as the last invariant factor 
of $\Sy$ without the use of 
  an additional 
condition with respect to the graded reverse lexicographic order~\cite{HoeLec21a}. 
This further allows us to deal with more general situations than that of the total degree since 
 we obtain the resultant in all cases where $\res_y(a,b)=c \mu$. This condition is sufficient but not necessary (\cref{example1,ex:notlemlaz7}), the resultant can be computed when 
$\Sy$ has a unique non-trivial invariant factor. Note that the  
latter property can be formalized in the Zariski sense, for example by relying on  
ideals in general position with no roots at infinity~\cite[Sec.\,3.5]{CoLiOSh05}. More precisely, there exists a non-zero polynomial $\Phi$
in $2(d+1)(e+1)$ variables over $\K$, such that the Smith form of $\Sy$ has a 
unique non-trivial invariant factor if $\Phi$ does not vanish at the coefficients of $a$ and $b$.

The generic resultant algorithm becomes of the Las Vegas type when the degree of the resultant is known in advance, 
especially if the Sylvester matrix $\Sy$ is column reduced. In the latter case the degree of the 
resultant is indeed the sum of the column degrees of $\Sy$ \cite[Eq. (24), p. 385]{Kailath80}.

%
%

\bibliographystyle{abbrvurl}

\vspace*{1cm}

{\small 
\renewcommand\bibname{References}
\let\clearpage\relax

\bibliography{ms}

\newcommand{\Hoeven}{\relax}\newcommand{\Gathen}{\relax}
\begin{thebibliography}{10}

\bibitem{Alonso96}
M.-E. Alonso, E.~Becker, M.-F. Roy, and T.~W{\"o}rmann.
\newblock \href {http://dx.doi.org/10.1007/978-3-0348-9104-2_1} {Zeros,
  multiplicities, and idempotents for zero-dimensional systems}.
\newblock In {\em Algorithms in Algebraic Geometry and Applications}, pages
  1--15. PM vol. 143, Birkha\"user, 1996.

\bibitem{BMMT94}
E.~Becker, T.~Mora, M.~G. Marinari, and C.~Traverso.
\newblock \href {http://dx.doi.org/10.1145/190347.190382} {{The Shape of the
  Shape Lemma}}.
\newblock In {\em Proc. ISSAC}, pages 129--133. ACM Press, 1994.

\bibitem{BeLa94}
B.~Beckermann and G.~Labahn.
\newblock \href {http://dx.doi.org/10.1137/S0895479892230031} {{A Uniform
  Approach for the Fast Computation of Matrix-Type Pad\'e Approximants}}.
\newblock {\em SIAM J. Matrix Analysis and Applications}, 15(3):804--823, 1994.

\bibitem{BNS22}
J.~Berthomieu, V.~Neiger, and M.~Safey El~Din.
\newblock \href {http://dx.doi.org/10.1145/3476446.3535484} {{Faster Change of
  Order Algorithm for {Gr\" obner} Bases under Shape and Stability
  Assumptions}}.
\newblock In {\em Proc. ISSAC}, pages 409--418. ACM Press, 2022.

\bibitem{BGGKU22}
V.~Bhargava, S.~Ghosh, Z.~Guo, M.~Kumar, and C.~Umans.
\newblock \href {http://dx.doi.org/10.1109/FOCS54457.2022.00028} {{Fast
  Multivariate Multipoint Evaluation Over All Finite Fields}}.
\newblock In {\em {Proc. FOCS}}, pages 221--232. IEEE, 2022.

\bibitem{BiPa94}
D.~Bini and V.~Y. Pan.
\newblock \href {http://dx.doi.org/10.1007/978-1-4612-0265-3} {{\em Polynomial
  and Matrix Computations}}.
\newblock Birkh\"auser, 1994.

\bibitem{Borde56}
J.~L. Bordewijk.
\newblock \href
  {https://repository.tudelft.nl/islandora/object/uuid:2a62abc0-cf96-4680-b72a-470d5a4ee502?collection=research}
  {{\em {Inter-reciprocity applied to electrical networks}}}.
\newblock PhD thesis, {Technische Hogeschool Delft}, 1956.

\bibitem{BFSS06}
A.~Bostan, P.~Flajolet, B.~Salvy, and E.~Schost.
\newblock \href {http://dx.doi.org/10.1016/j.jsc.2005.07.001} {Fast computation
  of special resultants}.
\newblock {\em J. Symb. Comput.}, 41(1):1--29, 2006.

\bibitem{BLS03}
A.~Bostan, G.~Lecerf, and {\'E}.~Schost.
\newblock \href {http://dx.doi.org/10.1145/860854.860870} {Tellegen's principle
  into practice}.
\newblock In {\em Proc. ISSAC}, pages 37--44. ACM Press, 2003.

\bibitem{BSS03}
A.~Bostan, B.~Salvy, and {\'E}.~Schost.
\newblock \href {http://dx.doi.org/10.1007/s00200-003-0133-5} {{Fast Algorithms
  for Zero-Dimensional Polynomial Systems using Duality}}.
\newblock {\em Appl. Algebr. Eng. Comm.}, 14:239--272, 2003.

\bibitem{BK78}
R.~P. Brent and H.~T. Kung.
\newblock \href {http://dx.doi.org/10.1145/322092.322099} {Fast algorithms for
  manipulating formal power series}.
\newblock {\em J. ACM}, 25(4):581--595, 1978.

\bibitem{BCS97}
P.~B{\"{u}}rgisser, M.~Clausen, and M.~A. Shokrollahi.
\newblock \href {http://dx.doi.org/10.1007/978-3-662-03338-8} {{\em Algebraic
  complexity theory}}, volume 315 of {\em Grundlehren der mathematischen
  Wissenschaften}.
\newblock Springer, 1997.

\bibitem{COS86}
D.~Coppersmith, A.~M. Odlzyko, and R.~Schroeppel.
\newblock \href {http://dx.doi.org/10.1007/BF01840433} {{Discrete logarithms in
  GF(p)}}.
\newblock {\em Algorithmica}, 1(1-4):1--15, 1986.

\bibitem{CouveignesLercier2013}
J.-M. Couveignes and R.~Lercier.
\newblock \href {http://dx.doi.org/10.1007/s11856-012-0070-8} {Fast
  construction of irreducible polynomials over finite fields}.
\newblock {\em Isr. J. Math.}, 194(1):77--105, 2013.

\bibitem{CD21}
D.~A. Cox and C.~D'Andrea.
\newblock \href {https://arxiv.org/abs/2112.10306} {{Subresultants and the
  Shape Lemma}}.
\newblock {arXiv:2112.10306}, 2021.

\bibitem{CoLiOSh05}
D.~A. Cox, J.~Little, and D.~O'Shea.
\newblock \href {http://dx.doi.org/10.1007/b138611} {{\em {Using Algebraic
  Geometry}}}.
\newblock Springer-Verlag, New-York, 1998.
\newblock 2nd edition 2005.

\bibitem{Dah22}
X.~Dahan.
\newblock \href {http://dx.doi.org/10.1016/j.jsc.2021.10.001} {Lexicographic
  {Gr\"obner} bases of bivariate polynomials modulo a univariate one}.
\newblock {\em J. Symb. Comput.}, 110:24--65, 2022.

\bibitem{faugere2014sub}
J.-C. Faug{\`e}re, P.~Gaudry, L.~Huot, and G.~Renault.
\newblock \href {http://dx.doi.org/10.1145/2608628.2608669} {{Sub-cubic change
  of ordering for Gr{\"o}bner basis: a probabilistic approach}}.
\newblock In {\em Proc. ISSAC}, pages 170--177. ACM Press, 2014.

\bibitem{FaMo11}
J.-C. Faug\`ere and C.~Mou.
\newblock \href {http://dx.doi.org/10.1145/1993886.1993908} {Fast algorithm for
  change of ordering of zero-dimensional {Gr\"obner} bases with sparse
  multiplication matrices}.
\newblock In {\em Proc. ISSAC}, pages 115--122. ACM Press, 2011.

\bibitem{FaMo17}
J.-C. Faug\`ere and C.~Mou.
\newblock \href {http://dx.doi.org/10.1016/j.jsc.2016.07.025} {Sparse {FGLM}
  algorithms}.
\newblock {\em J. Symb. Comput.}, 80(3):538--569, 2017.

\bibitem{Fidu73}
C.~M. Fiduccia.
\newblock \href
  {https://bruknow.library.brown.edu/permalink/01BU_INST/9mvq88/alma991027307479706966}
  {{\em On the algebraic complexity of matrix multiplication.}}
\newblock PhD thesis, Brown University, 1973.

\bibitem{GaGe99}
J.~\Gathen{von zur Gathen} and J.~Gerhard.
\newblock \href {http://dx.doi.org/10.1017/CBO9781139856065} {{\em Modern
  Computer Algebra}}.
\newblock Cambridge University Press, 1999.
\newblock Third edition 2013.

\bibitem{GTZ88}
P.~Gianni, B.~Trager, and G.~Zacharias.
\newblock \href {http://dx.doi.org/10.1016/S0747-7171(88)80040-3} {{Gr\"obner
  bases and primary decomposition of polynomial ideals}}.
\newblock {\em J. Symb. Computation}, 6(2-3):149--167, 1988.

\bibitem{Gonzalez-Vega1999}
L.~Gonzalez-Vega, F.~Rouillier, and M.-F. Roy.
\newblock \href {http://dx.doi.org/10.1007/978-3-662-03891-8_2} {{Symbolic
  Recipes for Polynomial System Solving}}.
\newblock In {\em Algorithms and Computation in Mathematics, Some Tapas of
  Computer Algebra}, pages 34--65. Springer, 1999.

\bibitem{Hoeven17}
J.~\Hoeven{van der Hoeven}.
\newblock \href {http://dx.doi.org/10.1007/978-3-319-56932-1_28} {{On the
  Complexity of Multivariate Polynomial Division}}.
\newblock In {\em Proc. ACA 2015}, pages 447--458. Springer, PROMS 198, 2017.

\bibitem{HoevenLarrieu2019}
J.~\Hoeven{van der Hoeven} and R.~Larrieu.
\newblock \href {http://dx.doi.org/10.1007/s00200-019-00389-9} {{Fast
  Gr{\"o}bner basis computation and polynomial reduction for generic bivariate
  ideals}}.
\newblock {\em Appl. Algebr. Eng. Comm.}, 30(6):509--539, 2019.

\bibitem{HoeLec20c}
J.~\Hoeven{van der Hoeven} and G.~Lecerf.
\newblock \href {http://dx.doi.org/10.1016/j.jco.2019.04.001} {Fast
  multivariate multi-point evaluation revisited}.
\newblock {\em J. of Complexity}, 56, 2020.

\bibitem{HoeLec21a}
J.~\Hoeven{van der Hoeven} and G.~Lecerf.
\newblock \href {http://dx.doi.org/10.1016/j.jco.2020.101499} {Fast computation
  of generic bivariate resultants}.
\newblock {\em J. of Complexity}, 62, 2021.

\bibitem{HoeLec21d}
J.~\Hoeven{van der Hoeven} and G.~Lecerf.
\newblock \href {http://dx.doi.org/10.1007/s10208-020-09453-0} {{On the
  Complexity Exponent of Polynomial System Solving}}.
\newblock {\em Found. Comput. Math.}, 21(1):1--57, 2021.

\bibitem{HMSS19}
S.~G. Hyun, S.~Melczer, {\'{E}}.~Schost, and C.~St{-}Pierre.
\newblock \href {http://dx.doi.org/10.1145/3326229.3326268} {Change of basis
  for m-primary ideals in one and two variables}.
\newblock In {\em Proc. ISSAC}, pages 227--234. ACM Press, 2019.

\bibitem{Jac85}
N.~Jacobson.
\newblock \href {https://store.doverpublications.com/0486471896.html} {{\em
  {Basic Algebra I}}}.
\newblock Dover Publications Inc., 2009.
\newblock Second Edition W.H. Freeman 1985.

\bibitem{Kailath80}
T.~Kailath.
\newblock {\em {Linear Systems}}.
\newblock Prentice-Hall, 1980.

\bibitem{KKM79}
T.~Kailath, S.~Y. Kung, and M.~Morf.
\newblock \href {http://dx.doi.org/10.1016/0022-247X(79)90124-0} {Displacement
  ranks of matrices and linear equations}.
\newblock {\em J. Math. Anal. Appl.}, 68(2):395--407, 1979.

\bibitem{Kal94}
E.~Kaltofen.
\newblock \href {http://dx.doi.org/10.1145/190347.190431} {{{Asymptotically
  fast solution of Toeplitz-like singular linear systems}}}.
\newblock In {\em Proc. ISSAC}, pages 297--304. ACM Press, 1994.

\bibitem{Kal00-2}
E.~Kaltofen.
\newblock \href {http://dx.doi.org/10.1006/jsco.2000.0370} {{Challenges of
  Symbolic Computation: My Favorite Open Problems}}.
\newblock {\em J. Symbolic Computation}, 29(6):891--919, 2000.

\bibitem{KL88}
E.~Kaltofen and Y.~Lakshman.
\newblock \href {http://dx.doi.org/10.1007/3-540-51084-2_44} {{Improved Sparse
  Multivariate Polynomial Interpolation Algorithms}}.
\newblock In {\em Proc. ISSAC}, pages 467--474. Springer, LNCS 358, 1988.

\bibitem{KaPa91}
E.~Kaltofen and V.~Y. Pan.
\newblock \href {http://dx.doi.org/10.1145/113379.113396} {Processor efficient
  parallel solution of linear systems over an abstract field}.
\newblock In {\em Proc. SPAA}, pages 180--191. {ACM}, 1991.

\bibitem{KaSa91}
E.~Kaltofen and D.~Saunders.
\newblock \href {http://dx.doi.org/10.1007/3-540-54522-0\_93} {On {W}iedemann's
  method of solving sparse linear systems}.
\newblock In {\em AAECC-9}, volume 539 of {\em LNCS}, pages 29--38. Springer
  Verlag, 1991.

\bibitem{KaSh98}
E.~Kaltofen and V.~Shoup.
\newblock \href {http://dx.doi.org/10.1090/S0025-5718-98-00944-2}
  {Subquadratic-time factoring of polynomials over finite fields}.
\newblock {\em Mathematics of Computation}, 67(233):1179--1197, 1998.

\bibitem{KU11}
K.~S. Kedlaya and C.~Umans.
\newblock \href {http://dx.doi.org/10.1137/08073408X} {{Fast Polynomial
  Factorization and Modular Composition}}.
\newblock {\em SIAM J. on Computing}, 40(6):1767--1802, 2011.

\bibitem{La1992}
G.~Labahn.
\newblock \href {http://dx.doi.org/10.1016/0024-3795(92)90316-3} {Inversion
  components of block {Hankel}-like matrices}.
\newblock {\em Linear Algebra Appl.}, 177:7--48, 1992.

\bibitem{Laz81}
D.~Lazard.
\newblock \href {http://dx.doi.org/10.1016/0304-3975(81)90064-5} {{R\'esolution
  des syst\`emes d'\'equations alg\'ebriques}}.
\newblock {\em Theor. Comput. Sci.}, 15(1):77--110, 1981.

\bibitem{Lazard85}
D.~Lazard.
\newblock \href {http://dx.doi.org/10.1016/S0747-7171(85)80035-3} {{Ideal Bases
  and Primary Decomposition: Case of Two Variables}}.
\newblock {\em J. Symb. Comput.}, 1(3):261--270, 1985.

\bibitem{lazard1992}
D.~Lazard.
\newblock \href {http://dx.doi.org/10.1016/S0747-7171(08)80086-7} {Solving
  zero-dimensional algebraic systems}.
\newblock {\em J. Symb. Comput.}, 13(2):117--131, 1992.

\bibitem{LMS13}
R.~Lebreton, E.~Mehrabi, and {\'E}.~Schost.
\newblock \href {http://dx.doi.org/10.1145/2465506.2465950} {On the complexity
  of solving bivariate systems: the case of non-singular solutions}.
\newblock In {\em Proc. ISSAC}, pages 251--258, 2013.

\bibitem{Lec17}
G.~Lecerf.
\newblock \href {https://hal.archives-ouvertes.fr/hal-01450869} {{On the
  complexity of the Lickteig-Roy subresultant algorithm}}.
\newblock {HAL report}, CNRS \& \'Ecole Polytechnique, 2017.

\bibitem{MS16b}
E.~Mehrabi and {\'E}.~Schost.
\newblock \href {http://dx.doi.org/10.1016/j.jco.2015.11.009} {{A softly
  optimal Monte Carlo algorithm for solving bivariate polynomial systems over
  the integers}}.
\newblock {\em J. of Complexity}, 34:78--128, 2016.

\bibitem{MoPa2000}
B.~Mourrain and V.~Y. Pan.
\newblock \href {http://dx.doi.org/10.1006/jcom.1999.0530} {{Multivariate
  Polynomials, Duality, and Structured matrices}}.
\newblock {\em J. of Complexity}, 16(1):110--180, 2000.

\bibitem{New72}
M.~Newman.
\newblock {\em Integral Matrices}.
\newblock Academic Press, 1972.
\newblock First edition.

\bibitem{pan01}
V.~Y. Pan.
\newblock \href {http://dx.doi.org/10.1007/978-1-4612-0129-8} {{\em Structured
  Matrices and Polynomials: Unified Superfast Algorithms}}.
\newblock Springer, 2001.

\bibitem{PaSc06}
C.~Pascal and {\'E}.~Schost.
\newblock \href {http://dx.doi.org/10.1145/1145768.1145814} {{Change of order
  for bivariate triangular sets}}.
\newblock In {\em Proc. ISSAC}, pages 277--284. ACM Press, 2006.

\bibitem{PSV22}
C.~Pernet, H.~Signargout, and G.~Villard.
\newblock \href {https://hal.science/hal-03740320v2/document} {{High-order
  lifting for polynomial Sylvester matrices}}.
\newblock {Hal-03740320}, 2022.

\bibitem{PS13}
A.~Poteaux and {\'{E}}.~Schost.
\newblock \href {http://dx.doi.org/10.1007/s00037-013-0063-y} {{Modular
  Composition Modulo Triangular Sets and Applications}}.
\newblock {\em Comput. Complex.}, 22(3):463--516, 2013.

\bibitem{RiBo91}
J.~Rif\`a and J.~Borrell.
\newblock \href {http://dx.doi.org/10.1007/3-540-54522-0_123} {Improving the
  time complexity of the computation of irreducible and primitive polynomials
  in finite fields}.
\newblock In {\em Proc. AAECC}, LNCS 539, pages 352--359, 1991.

\bibitem{Rou99}
F.~Rouillier.
\newblock \href {http://dx.doi.org/10.1007/s002000050114} {{Solving
  Zero-Dimensional Systems Through the Rational Univariate Representation}}.
\newblock {\em Appl. Algebra Eng. Commun. Comput.}, 9:433--461, 1999.

\bibitem{shoup91}
V.~Shoup.
\newblock \href {http://dx.doi.org/10.1145/120694.120697} {{A Fast
  Deterministic Algorithm for Factoring Polynomials over Finite Fields of Small
  Characteristic}}.
\newblock In {\em Proc. ISSAC}, pages 14--21, 1991.

\bibitem{Shoup94}
V.~Shoup.
\newblock \href {http://dx.doi.org/10.1006/jsco.1994.1025} {{Fast Construction
  of Irreducible Polynomials over Finite Fields}}.
\newblock {\em J. Symb. Comput.}, 17(5):371--391, 1994.

\bibitem{Shoup95}
V.~Shoup.
\newblock \href {http://dx.doi.org/10.1006/jsco.1995.1055} {{A new polynomial
  factorization algorithm and its implementation}}.
\newblock {\em J. Symb. Comput.}, 20(4):363--397, 1995.

\bibitem{Shoup99}
V.~Shoup.
\newblock \href {http://dx.doi.org/10.1145/309831.309859} {Efficient
  computation of minimal polynomials in algebraic extensions of finite fields}.
\newblock In {\em Proc. ISSAC}, pages 53--58. ACM Press, 1999.

\bibitem{Sto00}
A.~Storjohann.
\newblock \href
  {https://www.research-collection.ethz.ch/bitstream/handle/20.500.11850/145127/eth-24018-02.pdf?sequence=2}
  {{\em {Algorithms for Matrix Canonical Forms}}}.
\newblock PhD thesis, {Institut f\"ur Wissenschaftliches Rechnen, ETH-Zentrum,
  Zurich, Switzerland}, November 2000.

\bibitem{Thi89}
J.~A. Thiong~Ly.
\newblock \href {http://dx.doi.org/10.1007/BFb0019856} {Note for computing the
  minimun polynomial of elements in large finite fields}.
\newblock In {\em Proc. Coding Theo. App.}, LNCS 388, pages 185--192, 1989.

\bibitem{villard1997fast}
G.~Villard.
\newblock \href {http://dx.doi.org/10.1007/s002000050089} {{Fast Parallel
  Algorithms for Matrix Reduction to Normal Forms}}.
\newblock {\em Appl. Algebra Eng. Commun. Comput.}, 8(6):511--537, 1997.

\bibitem{Vil18}
G.~Villard.
\newblock \href {http://dx.doi.org/10.1145/3208976.3209020} {{On Computing the
  Resultant of Generic Bivariate Polynomials}}.
\newblock In {\em Proc. ISSAC}, pages 391--398. ACM Press, 2018.

\bibitem{Wie86}
D.~Wiedemann.
\newblock \href {http://dx.doi.org/10.1109/TIT.1986.1057137} {{Solving sparse
  linear equations over finite fields}}.
\newblock {\em IEEE Trans. Information Theory}, 32(1):54--62, 1986.

\end{thebibliography}
}

\end{document}